\documentclass[twoside,leqno,symbols-for-thanks]{rmi_modified}
\usepackage[colorlinks,linkcolor=black,citecolor=black,urlcolor=black, linktocpage=true,dvips]{hyperref}
\usepackage{srcltx}
\usepackage[english]{babel}
\numberwithin{equation}{section}
\setinitialpage{1}
\title[Domains for Dirac-Coulomb min-max levels]{Domains for Dirac-Coulomb min-max levels}

\author[M.J. Esteban, M. Lewin \& \'E. S\'er\'e]{Maria J. Esteban, Mathieu Lewin and \'Eric S\'er\'e}

\address[esteban@ceremade.dauphine.fr]{Maria J. Esteban: CEREMADE, CNRS, UMR 7534, Universit\'e Paris-Dauphine, PSL Research University, Place de Lattre de Tassigny, F-75016 Paris, France}
\address[mathieu.lewin@math.cnrs.fr]{Mathieu Lewin: CEREMADE, CNRS, UMR 7534, Universit\'e Paris-Dauphine, PSL Research University, Place de Lattre de Tassigny, F-75016 Paris, France}
\address[sere@ceremade.dauphine.fr]{\'Eric S\'er\'e: CEREMADE, Universit\'e Paris-Dauphine, PSL Research University, CNRS, UMR 7534, Place de Lattre de Tassigny, F-75016 Paris, France}

\amsclassification[35P05,81V45]{81Q10}


\usepackage{amsmath}
\usepackage{amssymb}
\usepackage{amsthm}
\usepackage{amsfonts}
\usepackage{amstext}
\usepackage{amsopn}
\usepackage{amsxtra}
\usepackage{mathrsfs}
\usepackage{dsfont}
\usepackage{esint}

\newtheorem{thm}{Theorem}[section]
\newtheorem{lemma}[thm]{Lemma}
\newtheorem{prop}[thm]{Proposition}
\newtheorem{cor}[thm]{Corollary}
\newtheorem{remark}[thm]{Remark}

\newcommand{\R}{\mathbb{R}}
\newcommand{\Z}{\mathbb{Z}}
\newcommand{\C}{\mathbb{C}}

\newcommand{\dps}{\displaystyle}
\newcommand{\ii}{\infty}
\newcommand\1{{\ensuremath {\mathds 1} }}
\renewcommand\phi{\varphi}
\newcommand{\gH}{\mathfrak{H}}

\newcommand{\wto}{\rightharpoonup}

\newcommand{\cV}{\mathcal{V}}

\newcommand{\cK}{\mathcal{K}}

\newcommand{\cD}{\mathcal{D}}
\newcommand{\cW}{\mathcal{W}}

\newcommand\pscal[1]{{\ensuremath{\left\langle #1 \right\rangle}}}
\newcommand{\norm}[1]{ \left\| #1 \right\|}

\renewcommand{\geq}{\geqslant}
\renewcommand{\leq}{\leqslant}

\renewcommand{\tilde}{\widetilde}

\newcommand{\be}{\begin{equation}}
\newcommand{\ee}{\end{equation}}
\newcommand{\bq}{\begin{equation}}
\newcommand{\eq}{\end{equation}}

\newcommand{\Frac}{\displaystyle \frac}

\newcommand{\Sup}{\displaystyle \sup}

\newcommand{\eps}{\varepsilon}
\usepackage{color}

\newcommand{\nn}{\nonumber}

\begin{document}

%
%

\begin{abstract}
We consider a Dirac operator in three space dimensions, with an electrostatic ({\it i.e.} real-valued) potential $V(x)$, having a strong Coulomb-type singularity at the origin. This operator is not always essentially self-adjoint but admits a distinguished self-adjoint extension $D_V$. In a first part we obtain new results on the domain of this extension, complementing previous works of Esteban and Loss. Then we prove the validity of min-max formulas for the eigenvalues in the spectral gap of $D_V$, in a range of simple function spaces independent of $V$. Our results include the critical case $\liminf_{x \to 0} |x| V(x)= -1$, with units such that $\hbar=mc^2=1$, and they are the first ones in this situation.  We also give the corresponding results in two dimensions.
\end{abstract}

%
%

\setcounter{tocdepth}{1}
\tableofcontents

\bigskip

Computing the eigenvalues in the gap of the essential spectrum of a self-adjoint operator is notoriously more difficult than for those below or above the essential spectrum. It is well-known that numerical artefacts can sometimes occur, a phenomenon called spectral pollution~\cite{LewSer-10}. For this reason, it is important to find robust methods.

In~\cite{EstSer-97,GriSie-99,DolEstSer-00a,DolEstSer-00}, variational min-max formulas were provided for the eigenvalues in gaps of self-adjoint operators. These formulas are based on a decomposition $\gH=\Lambda^+\gH\oplus\Lambda^-\gH$ given by two orthogonal projectors $\Lambda^\pm$ of the ambient Hilbert space $\gH$, and take the general form
\begin{equation}
 \lambda^{(k)}=\inf_{\substack{W\subset F^+\\ \dim(W)=k}}\sup_{\psi\in W\oplus F^-}\frac{\pscal{\psi,A\psi}}{\|\psi\|^2}.
 \label{eq:min-max_intro}
\end{equation}
Here, $F^\pm=\Lambda^\pm F$, with $F$ a dense subspace of $\gH$ such that the quadratic form $\pscal{\psi,A\psi}$ is well-defined on $F^+\oplus F^-\,.$

The equation~\eqref{eq:min-max_intro} is similar to the usual Courant-Fischer (a.k.a. Rayleigh-Ritz) formula for the eigenvalues below the essential spectrum. The main difference is that the infimum is restricted to vectors in the ``positive'' subspace $F^+$ and that the supremum is computed over the infinite-dimensional space $W\oplus F^-$ containing the whole ``negative'' space $F^-$. Some additional technical constraints on $F$ are needed, they are discussed in detail below. 

From the spectral theorem one can see that formula~\eqref{eq:min-max_intro} provides all the eigenvalues above a number $a'$ in the gap and below the next threshold of the essential spectrum, in nondecreasing order and counted with multiplicity, provided that we use for $\Lambda^-$ the spectral projector $\1(A\leq a')$ and, for instance, $F=\cD(A)$. Intuitively, formula~\eqref{eq:min-max_intro} should remain correct if $\Lambda^-$ is not too far from this spectral projector. The main discovery of~\cite{DolEstSer-00} was that the correct criterion for formula~\eqref{eq:min-max_intro} to provide the eigenvalues, is the inequality 
$$\lambda^{(1)}>a:=\sup_{\psi_-\in F^-}\frac{\pscal{\psi_-,A\psi_-}}{\|\psi_-\|^2}.$$
In practical cases, such a condition can be fulfilled for projectors $\Lambda^-$ which are quite far from the exact spectral projector $\1(A\leq a')$. Exploiting this freedom, one can choose $\Lambda^-$ so that the evaluation of the supremum in~\eqref{eq:min-max_intro} becomes very easy, leading to stable discretization techniques.

The main motivation for these min-max formulas was to study the spectrum of the free Dirac operator $D_0$ in 3d perturbed by an electrostatic potential $V$ with Coulomb-type singularity at the origin,
$$D_V=D_0+V(x).$$
The free Dirac operator $D_0$ in 3d is a constant-coefficient, first-order differential operator acting in $L^2(\R^3,\C^4)$ with spectrum $(-\ii,-1]\cup[1,\ii)$. Its precise definition and main properties are recalled below in Section~\ref{sec:free}. The potential $V$ is real-valued, bounded from above, and satisfies 
$$\liminf\limits_{x \to 0} |x| V(x)\geq -1$$ 
in units such that $\hbar=mc^2=1$. This class of operators is both important from the physical point of view and particularly challenging mathematically, due to the criticality of $1/|x|$ as compared with $D_0$. The first min-max formulas of the form \eqref{eq:min-max_intro} were proposed by Talman~\cite{Talman-86} and Datta-Devaiah~\cite{DatDev-88} in the particular case of the operators $D_V$, using the projectors $\Lambda^\pm$ associated with the natural decomposition
$$\Psi=\begin{pmatrix}\phi\\ \chi\end{pmatrix}=\begin{pmatrix}\phi\\ 0\end{pmatrix}+\begin{pmatrix}0\\ \chi\end{pmatrix}\in L^2(\R^3,\C^4),\qquad \phi,\chi\in L^2(\R^3,\C^2)$$
into upper and lower spinors. This choice leads to a particularly simple formula for the supremum in~\eqref{eq:min-max_intro}. It provides efficient ways of computing Dirac eigenvalues~\cite{DolEstSerVan-00,DolEstSer-03,KulKolRut-04,ZhaKulKol-04,CacDor-05}.

When dealing with unbounded quantum-mechanical operators, the questions of domain and self-adjointness are essential. These questions are delicate in the case of $D_V$ and have been the subject of an extensive literature: see, {\it e.g.},~\cite{Thaller,BalEva-11,Hogreve-13} and the references therein. For $0\leq\nu<\sqrt3/2$, if $V$ is real-valued and $\vert V(x)\vert\leq \nu/\vert x\vert$ then the minimal operator
$$\dot{D}_V:=(D_0+V)\upharpoonright C^\ii_c(\R^3\setminus\{0\},\C^4)$$
is essentially self-adjoint and the domain of its closure is $H^1(\R^3,\C^4)$. The minimal {\it exact} Dirac-Coulomb operator $\dot{D}_{-\nu/|x|}$ is still essentially self-adjoint\footnote{To our knowledge, essential self-adjointness is an open question for general real-valued potentials such that $\vert V(x)\vert \leq \frac{\sqrt{3}}{2|x|}$.} for $\nu=\sqrt{3}/2$, but it has infinitely many self-adjoint extensions for $\sqrt{3}/2<\nu\leq 1$. However, for any value $0\leq \nu< 1$, if $\vert V(x)\vert\leq \nu/\vert x\vert$ then the minimal operator $\dot{D}_V$ admits a distinguished self-adjoint extension $D_V$ with domain $\cD(D_V)$ characterized by the property $\cD(D_V)\subset H^{1/2}(\R^3,\C^4)$, which is the space on which the energy is well defined and continuous. The critical case $\nu=1$ is harder. It was considered for the first time by Esteban and Loss in~\cite{EstLos-07} who constructed a distinguished self-adjoint extension $D_V$ for real-valued potentials under the assumption $-1/|x|\leq V(x)\leq 0$. The properties of their extension will be discussed in detail in Section~\ref{sec:critical} below. 

As mentioned above, once the splitting $\gH=\gH^-\oplus\gH^+$ is chosen, one also has to choose the subspace $F$. In \cite{DolEstSer-00}, an abstract min-max theorem is proved, assuming that $F$ is a core (a dense subspace of $\cD(A)$ for the graph norm) and that $F^\pm$ are subspaces of $\cD(\vert A\vert^{1/2})$. In the application to Talman's principle when
$-\nu/\vert x\vert\leq V(x)\leq 0$ and $\nu<1$, a possible choice satisfying these requirements is $F=\cD(D_V)\subset H^{1/2}(\R^3,\C^4)$. But the domain $\cD(D_V)$ of the distinguished extension is not always explicitly known, so a natural question is whether the min-max can actually be performed on simpler spaces $F$ which do not depend on $V$. An attempt in this direction was made in~\cite{DolEstSer-00} where it was claimed that Talman's min-max formula holds for $F=C^\ii_c(\R^3,\C^4)$ as a consequence of the abstract theorem proved in the same paper. This was obvious for $0\leq \nu<\sqrt{3}/2$, indeed $\dot{D}_V$ is essentially self-adjoint, so $C^\ii_c(\R^3,\C^4)$ is a core. But the case $\sqrt{3}/2\leq\nu\leq 1$ was not properly justified in~\cite{DolEstSer-00}. An alternative approach was recently proposed by Morozov and Müller~\cite{MorMul-15,Muller-16}, who proved a variant of the abstract min-max formula allowing them to justify the choice $F=H^{1/2}(\R^3,\C^4)$ for any $\nu<1$.

In this paper we justify the application to $D_V$ of the abstract min-max of~\cite{DolEstSer-00}, for any subspace $F$ such that 
$$C^\ii_c(\R^3\setminus \{0\},\C^4)\;\subseteq\; F\;\subseteq\; H^{1/2}(\R^3,\C^4),$$ 
independently of the value of $0\leq\nu\leq 1$. In the critical case $\nu=1$ this provides the first min-max characterization of the eigenvalues. Our findings show that the min-max formula~\eqref{eq:min-max_intro} of the eigenvalues is valid for a wide range of spaces $F$, and is insensitive to the properties of the domain of the distinguished operator $D_V$. This is a clear advantage of this characterization, which fully justifies its use in practical computations. 

\medskip

In the first section we discuss domains of 3d Dirac-Coulomb operators with an emphasis on the distinguished self-adjoint extension. Most of the content of Sections~\ref{sec:free}--\ref{sec:general-potentials} is well known, and the results are presented here for the convenience of the reader. To our knowledge, the only novelty there is Proposition \ref{prop:closure}, which is proved in Appendix~\ref{app:proof-closure}. In Sections~\ref{sec:Esteban-Loss} and~\ref{sec:critical} we complement some results of Esteban-Loss~\cite{EstLos-07} on the characterization of the distinguished self-adjoint extension, using a quadratic form $q_E$ related to the min-max formula~\eqref{eq:min-max_intro}. Describing the domain of this quadratic form is important for knowing in which spaces the min-max can be formulated. In~\cite{EstLos-07} Esteban and Loss used the closure of $C^\ii_c$ for the norm induced by $q_E$. We show here that this coincides with the maximal domain on which the form $q_E$ is continuous. This is an important ingredient in our proof of the validity of the min-max formula. 

We also provide new results in the critical case $\nu=1$. In particular our proof that the resolvents converge in norm if the potential $V$ is truncated means that the Esteban-Loss extension is the only physically relevant extension for $\nu=1$. 

In Section~\ref{sec:min-max} we state our main result about the min-max formula that was claimed in~\cite{DolEstSer-00} and extend it to the critical case. Sections~\ref{sec:proof_thm_new},~\ref{sec:proof_quadratic form_domain},~\ref{sec:proof_thm_quadratic form_domain_critical},~\ref{sec:proof_thm_distinguished_critical} and Appendices~\ref{app:proof-closure}, \ref{sec:proof_thm_q_C} are dedicated to the proof of our results.  Our results are stated and proved in detail in three space dimensions, but they can easily be adapted to the two-dimensional setting. This is explained in Appendix~\ref{sec:2D}.

\section{Domains of Dirac-Coulomb operators in 3d}

In this section we discuss domains for Dirac-Coulomb operators in three space dimensions, and provide some new properties of the distinguished self-adjoint extension. Some of these properties will be useful in Section~\ref{sec:min-max} where we prove the min-max formula for the eigenvalues.

\subsection{The free Dirac operator in 3d}\label{sec:free}
In a system of units such that $\hbar=m=c=1$, the free Dirac operator $D_0$ in 3d is given by
\begin{equation}
D_0\ = -i\; \boldsymbol{\alpha}\cdot\boldsymbol{\nabla} + \beta = \ - i\
\sum^3_{k=1} {\bf \alpha}_k \partial _k + {\bf \beta},
\label{def_Dirac}
\end{equation}
where $\alpha_1$, $\alpha_2$, $\alpha_3$ and $\beta$ are $4\times4$ Hermitian matrices satisfying the anticommutation relations
\begin{equation} \label{CAR}
\left\lbrace
\begin{array}{rcl}
 {\alpha}_k
{\alpha}_\ell + {\alpha}_\ell
{\alpha}_k  & = &  2\,\delta_{k\ell}\,\1_{\C^4},\\
 {\alpha}_k {\beta} + {\beta} {\alpha}_k
& = & 0,\\
\beta^2 & = & \1_{\C^4}.
\end{array} \right. \end{equation}
The usual representation in $2\times 2$ blocks is given by 
$$ \beta=\left( \begin{matrix} I_2 & 0 \\ 0 & -I_2 \\ \end{matrix} \right),\quad \; \alpha_k=\left( \begin{matrix}
0 &\sigma_k \\ \sigma_k &0 \\ \end{matrix}\right),  \qquad k=1, 2, 3\,,
$$
with the Pauli matrices
$$\sigma _1=\left( \begin{matrix} 0 & 1
\\ 1 & 0 \\ \end{matrix} \right),\quad  \sigma_2=\left( \begin{matrix} 0 & -i \\
i & 0 \\  \end{matrix}\right),\quad  \sigma_3=\left( 
\begin{matrix} 1 & 0\\  0 &-1\\  \end{matrix}\right) \, .$$
The operator $D_0$ is self-adjoint in $L^2(\R^3,\C^4)$ with domain $H^1(\R^3,\C^4)$ and its spectrum is $\sigma(D_0)=(-\ii,-1]\cup[1,\ii)$, see~\cite{Thaller,EstLewSer-08}. 
In addition, the corresponding quadratic form $\pscal{\Psi,D_0\Psi}$ is well-defined and continuous on the Sobolev space $H^{1/2}(\R^3,\C^4)$, which is also the domain of $|D_0|^{1/2}=(1-\Delta)^{1/4}$.

The Rellich-Kato theorem and the Sobolev inequality imply that 
$$D_V:=D_0+V(x)$$
is also self-adjoint on $H^1(\R^3,\C^4)$ for any real-valued potential $V\in L^3(\R^3,\R)+L^\ii(\R^3,\R)$. The purpose of this article is to discuss the case of Coulomb-type potentials which behave like $-\nu|x|^{-1}$ near to the origin, and which just fail to be in $L^3$ at the origin. Using Hardy's inequality
$$\frac1{|x|^2}\leq 4(-\Delta)\leq 4(D_0)^2=4(-\Delta+1)$$
we can use again the Rellich-Kato theorem and obtain that $D_V$ is self-adjoint on $H^1(\R^3,\C^4)$ for potentials in the form $V=V_1+V_2$ where $V_2\in L^3(\R^3,\R)+L^\ii(\R^3,\R)$ and $|V_1(x)|\leq \nu|x|^{-1}$ with $|\nu|<1/2$. However, the threshold $1/2$ given by this argument is not optimal and the proper limit is, rather, $\sqrt{3}/2$ (at least for scalar potentials, see Remark~\ref{rmk:matrix_V} below for matrix potentials). In order to understand the situation, it is enlightening to first look at the well-known case of the exact Coulomb potential.

\subsection{The exact Coulomb potential}\label{sec:radial}
Here we discuss the well-known exact Coulomb case 
$$V_{\rm C}(x)=-\frac{\nu}{|x|}.$$
Note that when $V$ is a bounded perturbation of this potential, the self-adjoint realizations of $D_V$ have the same domains as for $V_{\rm C}$.

For a radial potential such as $V_{\rm C}$, one can use that the Dirac operator commutes with the total angular momentum $J=L+S=(J_1,J_2,J_3)$, as well as with the spin-orbit operator $K=\beta(2S\cdot L+1)$, see~\cite[Sec.~4.6]{Thaller}. Viewing $K$, $J_3$ and $\beta$ as a complete set of commuting observables in the Hilbert space $L^2(\mathbb{S}^2,\C^4)$, one finds an orthonormal basis of this space consisting of trigonometric polynomials in the spherical coordinates $(\theta,\phi)$,
$\Phi_{\kappa,m}^\pm$, indexed by $\kappa\in\Z\setminus\{0\}$ and $m\in\{-\vert \kappa\vert+1/2,-\vert \kappa\vert+3/2,\cdots,\vert \kappa\vert-1/2\}$. 
Using this basis, for any $\Psi\in\C^\infty_c(\R^3\setminus\{0\},\C^4)$ we get the $L^2$-orthogonal decompositions
\begin{equation}
\Psi(x)=r^{-1}\sum_{\kappa,m}u_{\kappa,m}(r)\Phi_{\kappa,m}^+(\theta,\phi)+v_{\kappa,m}(r)\Phi_{\kappa,m}^-(\theta,\phi)\,,
\label{splitting}
\end{equation}
and
\begin{equation}
D_{-\nu/r} \Psi (x) = r^{-1}\sum_{\kappa,\,m}f_{\kappa,m}(r)\Phi_{\kappa,m}^+(\theta,\phi)+g_{\kappa,m}(r)\Phi_{\kappa,m}^-(\theta,\phi),
\label{splittingDVPsi}
\end{equation}
where $r=|x|$, $K\Phi_{\kappa,m}^\pm=-\kappa\Phi_{\kappa,m}^\pm\,,$ $J_3\Phi_{\kappa,m}^\pm=m\Phi_{\kappa,m}^\pm\,,$ $\beta\Phi_{\kappa,m}^\pm=\pm\Phi_{\kappa,m}^\pm\,$ and
\begin{equation}
\begin{pmatrix}f_{\kappa,m}\\ g_{\kappa,m}\end{pmatrix}=h^{\kappa}_\nu\begin{pmatrix}u_{\kappa,m}\\ v_{\kappa,m}\end{pmatrix} ,
\label{eq:DV_radial}
\end{equation}
and where we have introduced the radial Coulomb Dirac-type operator 
\begin{equation} \label{radial}
h^{\kappa}_\nu=\begin{pmatrix}
1-\frac\nu{r} & -\frac{d}{dr}+\frac{\kappa}{r}\\
\frac{d}{dr}+\frac{\kappa}{r} & -1-\frac\nu{r}
\end{pmatrix}.
\end{equation}
As a consequence, the Dirac operator $D_{{-\nu/r}}$ is unitarily equivalent to the direct sum (with multiplicities $2\vert\kappa\vert-1$) of the radial Dirac-type operators $h^{\kappa}_\nu$ acting in the Hilbert space $L^2((0,\ii),\C^2)$.
Using ODE techniques, the question of self-adjointness is then reduced to the discussion of the possible boundary conditions at $r=0$, see~\cite{Rellich-43,Case-50,Weidmann-87,Evans-70,VorGitTyu-07,Hogreve-13,Thaller}. 

Let us discuss this in more detail. In order to find the self-adjoint extensions of the minimal operator
$$\dot{h}_\nu^{\kappa}:=h^{\kappa}_\nu \upharpoonright\C^\ii_c((0,\ii),\C^2)\,,$$
we compute its deficiency subspaces\footnote{Here we follow von Neumann's theory of self-adjoint extensions~\cite{vonNeumann-30}. We refer to the recent paper~\cite{GalMic-19} for an alternative approach based on the Kre{\u\i}n-Vi\v{s}ik-Birman extension scheme.} ${\mathcal K}_\pm=\ker\big((\dot{h}_\nu^{\kappa})^*\mp i\big)$. Since $\cK_-=\overline{\cK_+}$, we only have to determine ${\mathcal K}_+$. The corresponding eigenvalue equation is
\begin{equation}
\begin{cases}
(1-\nu/r)u-v'+\frac{\kappa}r v=iu,\\
u'+\frac{\kappa}ru-(1+\nu/r)v=iv.
\end{cases}
\label{eq:system_radial}
\end{equation}
Plugging in the first equation the relation
$$v=\frac{u'+\frac\kappa{r}u}{1+\nu/r+i}$$
deduced from the second one (note that the denominator never vanishes), 
we obtain an equation for $u$ only:
\begin{equation}
\left(-\frac{d}{dr}+\frac{\kappa}r\right)\frac{1}{1+\nu/r+i}\left(\frac{d}{dr}+\frac{\kappa}r\right)u+ \left(1-\frac\nu{r}-i\right)u=0.
\label{eq:equation_u_only}
\end{equation}
Using standard ODE techniques, one finds that the solution space of~\eqref{eq:equation_u_only} is spanned by two independent functions behaving as $r^{\pm s}(1+O(r))$ at $r=0$, with $s:=\sqrt{\kappa^2-\nu^2}$. Another basis of this space consists of two independent solutions behaving like 
$$\exp(\pm\sqrt2 r)r^{\mp i\nu/\sqrt2}(1+O(r^{-1}))$$ 
when $r\to\ii$~\cite{Plesset-32,Case-50,Titchmarsh-61}. At $\nu=1$, for $\kappa=\pm 1$ we have $s=0$ and there are two solutions behaving like $1$ and $\log(r)$, respectively, near $r=0$.

We first assume $|\nu|<1$. The solution $u_+$ which behaves like $r^s$ at $0$ must diverge at infinity, hence is not in $L^2$. Indeed, assuming by contradiction that $u_+$ behaves as $\exp(-\sqrt2 r)r^{i\nu/\sqrt2}$ at infinity, we can multiply~\eqref{eq:equation_u_only} by $\overline{u_+}$ and integrate by parts (the boundary terms cancel due to the behavior at the origin and at infinity), which gives
$$\int_0^\ii \frac{|u_+'(r)+\kappa u_+(r)/r|^2}{1+\nu/r+i}\,dr=\int_0^\ii (i-1+\nu/r)|u_+(r)|^2\,dr.$$
The imaginary part is negative for the first term and positive for the second, which is a contradiction. 

The solution $u_-$ which behaves like $r^{-s}$ is not square-integrable at the origin when $|\nu|\leq\sqrt{\kappa^2-1/4}$. The smallest value of this threshold is $\sqrt{3}/2$ which we have mentioned before, and it is obtained for $\kappa=\pm1$. 
We conclude that the deficiency indices $n_\pm=\dim {\mathcal K}_\pm$ vanish for $|\nu|\leq\sqrt{3}/2$ and that the operator is essentially self-adjoint in this case. When $|\nu|<\sqrt{3}/2$, the domain of the closure of $\dot{h}_\nu^\kappa$ 
can be shown to be $H^1_0((0,\ii),\C^2)$, and that of $\dot{D}_{-\nu/r}$ to be 
$$\cD(D_{-\nu/r})=H^1(\R^3,\C^4),\qquad \text{for $|\nu|<\sqrt{3}/2$,}$$ 
see~\cite{LanRej-79,LanRejKla-80}.
The situation is more complicated at $|\nu|=\sqrt3/2$. Although the operator is essentially self-adjoint, its domain is larger than $H^1(\R^3,\C^4)$:
$$\cD(D_{-\nu/r}) \supsetneq H^1(\R^3,\C^4),\qquad \text{for $|\nu|=\sqrt{3}/2$.}$$ 
We explain all this in Proposition~\ref{prop:closure} of Appendix~\ref{app:proof-closure}.

When $\sqrt{3}/2< |\nu|\leq 1$ the arguments of~\cite{LanRej-79,LanRejKla-80} apply for $|\kappa|\geq 2$ and show that the operators $\dot{h}_\nu^{\kappa}$ are all essentially self-adjoint, with domain
$$\cD\big(\overline{\dot{h}_\nu^\kappa}\big)=H^1_0\big((0,\ii),\C^2\big),\qquad \text{for $|\kappa|\geq 2$ and $\sqrt{3}/2< |\nu|\leq 1$.}$$
Only $\kappa=\pm1$ pose some difficulties. 

In the case $\sqrt{3}/2< |\nu|<1$, the two functions $r^{\pm \sqrt{1-\nu^2}}$ are now square-integrable at 0 and there is one linear combination of $u_+$ and $u_-$, which we call $u_\kappa$, which is square-integrable at infinity. From the previous argument, this function must diverge like $r^{-s}$ at $r=0$ and we can therefore always assume that $u_\kappa\sim r^{-s}$ and $v_\kappa\sim (\kappa-s)r^{-s}/\nu$ at 0. By von Neumann's theory of self-adjoint extensions (see, {\it e.g.}, \cite[p.~140]{ReeSim2}), we conclude that, for $\kappa=\pm1$, $\dot{h}_\nu^{\kappa}$ admits a family of self-adjoint extensions parametrized by $\alpha\in[0,2\pi)$, whose domains are given by 
$$\cD\left(\overline{\dot{h}_\nu^{\kappa}}\right)\oplus \left\{\begin{pmatrix}u_\kappa\\ v_\kappa\end{pmatrix}+e^{i\alpha}\begin{pmatrix}\overline{u_\kappa}\\ \overline{v_\kappa}\end{pmatrix}\right\}\C,\qquad \text{for $\kappa=\pm1$.}$$
In this formula, we have used that the solutions with eigenvalue $+i$ and $-i$ are related by complex conjugation, since the operator $h_\nu^{\kappa}$ is real. In Proposition~\ref{prop:closure} in Appendix \ref{app:proof-closure} we will prove that for all $\sqrt{3}/2< |\nu|\leq1$,
$$\cD\left(\overline{\dot{h}_\nu^{\pm1}}\right)=H^1_0((0,\infty),\C^2)\,.$$
The functions $(u_\kappa^\alpha,v_\kappa^\alpha)=(u_\kappa+e^{i\alpha}\overline{u_\kappa},v_\kappa+e^{i\alpha}\overline{v_\kappa})$, $\kappa=\pm1$, 
are more singular at the origin. For $\alpha\neq\pi$, they have the strong singularity $(1+e^{i\alpha})r^{-\sqrt{1-\nu^2}}$ at $r=0$. The associated $4$-spinors
$$\Psi^\alpha_{\kappa,m}:=\vert x\vert^{-1}\left(u_\kappa^\alpha(r)\Phi_{\kappa,m}^+(\theta,\phi)+v_\kappa^\alpha(r)\Phi_{\kappa,m}^-(\theta,\phi)\right)\,,\; m=\pm1/2\,,$$
will not have a finite Coulomb energy and will not be in $H^{1/2}(\R^3,\C^4)$ (the natural space for which one can define the quadratic form of the free Dirac operator). However, if we choose $\alpha=\pi$, the function behaves like 
\begin{align*}
u_\kappa-\overline{u_\kappa}&=2i\Im\big\{r^{-\sqrt{1-\nu^2}}(1+O(r))+ar^{\sqrt{1-\nu^2}}(1+O(r))\big\}\\
&=2i\Im(a) r^{\sqrt{1-\nu^2}}+O(r^{1-{\sqrt{1-\nu^2}}})                                                                                                                                                                                                                                                                                                                                                                                                                                                                                                                                                                                                                                                                                                                                                                                            \end{align*}
as $r\to 0$, since $1/2>{\sqrt{1-\nu^2}}$. Therefore the associated $4$-spinor has a finite Coulomb energy as well as a well-defined free Dirac energy. This sounds more satisfactory from a physical point of view. Note however that $r^{\sqrt{1-\nu^2}}$ is not in $H^1$ at the origin for $\sqrt{3}/2<|\nu|<1$, hence the domain of this self-adjoint realization is always bigger than $H^1$.

The realization of the Dirac operator which has $\alpha=\pi$ in the four sectors corresponding with the quantum numbers $\kappa=\pm1$ and $m=\pm1/2$ is called the \emph{distinguished self-adjoint extension} of the minimal Dirac-Coulomb operator $\dot{D}_{-\nu/r}$.

For $\nu=\pm1$ the situation is slightly different since $s=0$. The two functions behave at the origin like $1$ and $\log(r)$. Hence even for $\alpha=\pi$, the Coulomb energy is infinite since $u_{\pm1}$ does not tend to 0 at 0. 
However it can be called a distinguished extension since it is the least singular. It can also be shown that it is the one obtained when $\nu\to\pm1^\mp$, as we will discuss for general potentials in Section~\ref{sec:critical}, and the one for which the min-max characterization holds in any reasonable space that one can think of. 

If we now come back to the whole space and use~\cite[Sec.~4.6.4]{Thaller}, the corresponding domain of the distinguished extension reads
\begin{multline}
\label{eq:domain_distinguished_radial}
\cD(D_{-\nu/r})=\cD\left(\overline{\dot{D}_{-\nu/r}}\right)\oplus
\begin{pmatrix} 
i\,U_{-1}\begin{pmatrix}  Y^0_0 \\0 \\  \end{pmatrix} \\  \\  \frac{V_{-1}}{\sqrt{3}}\left( \begin{matrix} Y^0_1 \\-\sqrt{2}\,Y^1_1 \\   \end{matrix} \right)  \\
\end{pmatrix} 
\oplus 
\begin{pmatrix} 
i\,U_{-1}\left( \begin{matrix} 0\\ Y^0_0  \\  \end{matrix} \right) \\  \\  \frac{V_{-1}}{\sqrt{3}}\left( \begin{matrix} \sqrt{2}\,Y^{-1}_1 \\-Y^0_1 \\   \end{matrix} \right)\end{pmatrix}\\
\oplus\begin{pmatrix} 
i\,\frac{U_1}{\sqrt{3}}\left( \begin{matrix} Y^0_1 \\-\sqrt{2}\,Y^1_1 \\   \end{matrix} \right) \\  \\  V_1\left( \begin{matrix}  Y^0_0 \\0 \\  \end{matrix} \right)  \\
\end{pmatrix} 
\oplus \begin{pmatrix} 
i\,\frac{U_1}{\sqrt{3}}\left( \begin{matrix} \sqrt{2}\,Y^{-1}_1 \\-Y^0_1 \\   \end{matrix} \right) \\  \\  V_1\left( \begin{matrix} 0\\ Y^0_0  \\  \end{matrix} \right)  \\
\end{pmatrix},
\end{multline}
where  the functions $Y^m_\ell$ are the spherical harmonics normalized as in \cite[Sec.~4.6.4]{Thaller} and $U_\kappa=(u_\kappa-\overline{u_\kappa})/r$, $V_\kappa=(v_\kappa-\overline{v_\kappa})/r$.  Moreover, we prove in Appendix \ref{app:proof-closure} that
\begin{equation}
 \cD\left(\overline{\dot{D}_{-\nu/r}}\right)=H^1(\R^3,\C^4)\,,\qquad\text{for $\sqrt3/2<|\nu|\leq1$.}
 \label{eq:equality_H_1}
\end{equation}
We conclude that, for $\sqrt3/2<|\nu|\leq1$, the domain of the distinguished self-adjoint extension is just the usual Sobolev space $H^1(\R^3,\C^4)$ to which are added four functions having an explicit singularity at the origin, which is so strong that these are always outside of $H^1(\R^3,\C^4)$. They belong to $H^{1/2}(\R^3,\C^4)$ when $\sqrt3/2<|\nu|<1$, but just fail to do so when $|\nu|=1$.

\subsection{General potentials with subcritical Coulomb-like singularity}\label{sec:general-potentials}
It is natural to ask whether similar results hold for potentials which have a singularity that can be controlled in absolute value by $\nu|x|^{-1}$ without being a bounded perturbation of $\pm\nu/|x|$. In the seventies and eighties, many authors~\cite{Schmincke-72b,Wust-73,Wust-75,Wust-77,Nenciu-76,KlaWus-78,LanRej-79,LanRejKla-80,Kato-83,Thaller} have proved the existence of a distinguished self-adjoint extension when $|\nu|<1$ which has the same properties as in the exact Coulomb case. The following statement is a summary of several of these results, some of which will be useful for us later.

\begin{thm}[Distinguished extension of $\dot{D}_V$~\cite{Schmincke-72b,Wust-73,Wust-75,Wust-77,Nenciu-76,KlaWus-78,LanRej-79,LanRejKla-80,Kato-83,Thaller}]\label{thm:distinguished}
We assume that $V=V_1+V_2+V_3$ with $V_2\in L^3(\R^3,\R)$, $V_3\in L^\ii(\R^3,\R)$ and $|V_1(x)|\leq \nu/|x|$, with $0\leq \nu<1$.

\begin{enumerate}
 \item The minimal operator $\dot{D}_V$ defined on $C^\ii_c(\R^3\setminus\{0\},\C^4)$ has a unique self-adjoint extension $D_V$ such that 
$$H^1(\R^3,\C^4)\subset \cD(D_V)\subset H^{1/2}(\R^3,\C^4).$$
It is also the unique self-adjoint extension for which 
$$\int_{\R^3}\frac{|\Psi(x)|^2}{|x|}\,dx<\ii,\qquad \forall \Psi\in \cD(D_V).$$
\item For any $\Psi,\Psi'\in\cD(D_V)$, we have
\begin{equation}
\pscal{\Psi,D_V\Psi'}=\pscal{\Psi,D_0\Psi'}+\int_{\R^3}V\Psi^*\Psi'
\label{eq:form}
\end{equation}
where the right-side is understood in the form sense in $H^{1/2}(\R^3,\C^4)$.
\item If $V_3\to0$ at infinity, the essential spectrum is 
$$\sigma_{\rm ess}(D_V)=(-\ii,-1]\cup[1,\ii).$$
\item For $V_\eps:=\min(\max(V(x),-1/\eps),1/\eps)$, the operator $D_{V_\eps}$ converges to the distinguished self-adjoint extension in the norm resolvent sense when $\eps\to0$.
\item If in addition $0\leq \nu<\sqrt{3}/2$, then the operator $\dot{D}_V$ is essentially self-adjoint on $C^\ii_c(\R^3\setminus\{0\},\C^4)$ and its domain is $\cD(D_V)=H^1(\R^3,\C^4)$.
\end{enumerate}
\end{thm}

\begin{remark}\label{rmk:form}
We have $H^1(\R^3,\C^4)=\cD(|D_0|)\subset\cD(|D_V|)$. Since the square root is operator monotone, we deduce that 
$$\cD(|D_0|^{1/2})=H^{1/2}(\R^3,\C^4)\subset \cD(|D_V|^{1/2}).$$ 
This can be used to extend the formula~\eqref{eq:form} on the whole of $H^{1/2}(\R^3,\C^4)$, if we interpret the left side in the sense of quadratic forms, that is, $\pscal{\Psi,D_V\Psi'}:=\pscal{|D_V|^{1/2}\Psi,U_V|D_V|^{1/2}\Psi'}$ where $U_V={\rm sgn}(D_V)$.
\end{remark}

\begin{remark}\label{rmk:matrix_V}
The results are exactly the same for a Hermitian $4\times4$ matrix potential $V(x)$, with the exception of \textit{(5)} in which $\sqrt3/2$ has to be replaced by $1/2$. There are examples of matrix-valued potentials satisfying $|V(x)|\leq(1+\eps)/(2|x|)$ for which $D_V$ is not essentially self-adjoint~\cite{Arai-75}.
\end{remark}

In~\cite{Nenciu-76}, Nenciu  defines the distinguished self-adjoint extension through its resolvent, using the formula
\begin{equation}
\frac{1}{D_V-z}=\frac{1}{D_0-z}- \frac{1}{D_0-z}|V|^{\frac12}\frac{1}{1+SM(z)}|V|^{\frac12}\frac{1}{D_0-z}
\label{eq:resolvent}
\end{equation}
where $S={\rm sgn}(V)$ and $M(z)=|V|^{1/2}(D_0-z)^{-1}|V|^{1/2}$.
From Kato's inequality 
\begin{equation}
\frac{1}{|x|}\leq \frac\pi2 \sqrt{-\Delta}\leq \frac\pi2|D^0|
\label{eq:Kato}
\end{equation}
and Sobolev's inequality, one can prove that $|V|^{1/2}|D_0|^{-1/2}$ and $|D_0|^{-1/2}|V|^{1/2}$ are bounded under the assumptions of Theorem~\ref{thm:distinguished}. Then $(D_0-z)^{-1}|V|^{1/2}$ (appearing on the left of the last term in~\eqref{eq:resolvent}) has its range in $H^{1/2}(\R^3,\C^4)$. This shows that the range of $(D_V-z)^{-1}$ (that is, the domain of $D_V$) is included in $H^{1/2}(\R^3,\C^4)$, as required. In addition, since $(D_0-z)^{-1}|V|^{1/2}$ is compact under our assumptions on $V$ by~\cite[Sec.~5.7]{Davies-07}, $(D_V-z)^{-1}$ is a compact perturbation of $(D_0-z)^{-1}$, and the two operators have the same essential spectrum~\cite{KlaWus-79}.

The main condition necessary to give a meaning to~\eqref{eq:resolvent} is that $1+SM(z)$ is invertible on $L^2(\R^3,\C^4)$. Nenciu proves that $D_V$ is uniquely defined from~\eqref{eq:resolvent} under the sole condition that $\|M(z_0)\|<1$ for one $z_0\in\C$. Since $z\mapsto (1+SM(z))^{-1}$ is meromorphic on $\C\setminus(-\ii,-1]\cup[1,\ii)$, this is sufficient to define the right side of~\eqref{eq:resolvent} for a large set of values of $z$, and then to construct the operator $D_V$.
In our case the bound on $M(z_0)$ follows from the two equalities
\begin{equation}
\norm{|x|^{-1/2}(D_0+is)^{-1}|x|^{-1/2}}=1,\qquad \forall s\in\R,
\label{eq:Nenciu-Kato}
\end{equation}
and
\begin{equation}
 \lim_{s\to\ii}\norm{|V_2|^{1/2}(D_0+is)^{-1}|V_2|^{1/2}}=0\quad\text{for $V_2\in L^{3}(\R^3)$}.
 \label{eq:limit_Sobolev}
\end{equation}
The limit~\eqref{eq:limit_Sobolev} follows from the Sobolev inequality. The equality~\eqref{eq:Nenciu-Kato} was conjectured by Nenciu in~\cite{Nenciu-76} and later proved by Wüst~\cite{Wust-77} and Kato~\cite{Kato-83}. It has recently been rediscovered in~\cite[Thm.~1.3]{ArrDuoVeg-13}. The constraint that $|\nu|<1$ comes from the norm in~\eqref{eq:Nenciu-Kato} being equal to 1.

\subsection{A different characterization of the distinguished extension}\label{sec:Esteban-Loss}

Now we turn to the description of a method which has been introduced in~\cite{EstLos-07,EstLos-08} (further developed in~\cite{Arrizabalaga-11,ArrMasVeg-14,ArrMasVeg-15}), and is essential for our discussion of min-max levels. We are going to make the stronger assumption 
\begin{equation}
 \boxed{-\frac{\nu}{|x|}\leq V(x)<1+\sqrt{1-\nu^2}} 
 \label{eq:assumption_V_intro}
\end{equation}
for some $0\leq \nu<1$. Here $\sqrt{1-\nu^2}$ is the first eigenvalue of the Dirac operator with the Coulomb potential $V_{\rm C}(x)=-\nu/|x|$.
The lower bound in~\eqref{eq:assumption_V_intro} means that the attractive part of $V$ is essentially Coulombic and it will imply that the first ``electronic'' eigenvalue will be above $\sqrt{1-\nu^2}$. Here ``electronic'' means that it is an eigenvalue which arises from the upper part of the spectrum when $V$ is replaced by $tV$ and $t$ is turned on progressively. The upper bound on $V$ in~\eqref{eq:assumption_V_intro} is here to ensure that the positronic eigenvalues (those arising from the lower part) do not go above $\sqrt{1-\nu^2}$. The fact that the electronic and positronic eigenvalues do not cross is an important property for having a min-max formula of the eigenvalues (see~\cite{DolEstSer-06} for a discussion). 

In this section we introduce a quadratic form for the upper spinor, which plays a central role in the definition of the distinguished self-adjoint extension and for the min-max formulation of the electronic eigenvalues.

Similarly as in Subsection~\ref{sec:radial}, we consider the eigenvalue equation $D_V\Psi=\lambda\Psi$ with, this time, $\lambda\in\R$, and which we write in terms of the upper and lower components $\phi,\chi\in L^2(\R^3,\C^2)$ of the 4--spinor $\Psi=\begin{pmatrix}\phi\\ \chi\end{pmatrix}$. We obtain 
\begin{equation}
\begin{cases}
(1+V)\phi-i\sigma\cdot\nabla \chi=\lambda\phi,\\
-i\sigma\cdot\nabla \phi+(-1+V)\chi=\lambda\chi.
\end{cases}
\end{equation}
We insert 
$$\chi=\frac{-i\sigma\cdot\nabla\phi}{1-V+\lambda}$$
in the first equation and get an equation for $\phi$ only:
\begin{equation}
 -i\sigma\cdot\nabla \frac{-i\sigma\cdot\nabla\phi}{1-V+\lambda}+(1+V-\lambda)\phi=0. 
 \label{eq:phi}
\end{equation}
This suggest to look at the quadratic form
\begin{equation}
q_\lambda(\phi):=\int_{\R^3}\frac{|\sigma\cdot \nabla\phi(x)|^2}{1-V(x)+\lambda}\,dx+\int_{\R^3}(1+V(x)-\lambda)|\phi(x)|^2\,dx.
 \label{eq:q_E}
\end{equation}
Note that the denominator in the first term is well defined for $\lambda>\sup(V)-1$. Without $\lambda$ in the denominator of the first term, which comes from the lower component $\chi$, the quadratic form $q_\lambda$ would be associated with a usual eigenvalue problem. With $\lambda$ in the denominator this is more involved. Nevertheless we have gained that the solution $\phi$ to~\eqref{eq:phi} can be constructed by a minimization procedure, for any $\lambda>\sup(V)-1$.
In Section~\ref{sec:min-max} we will explain the link between the quadratic form $q_\lambda$ and the true eigenvalues of $D_V$ but, for the moment, we discuss the properties of $q_\lambda$ for an arbitrary $\lambda>\sup(V)-1$.

In order to show that $q_\lambda$ is bounded from below, we write
\begin{align*}
q_\lambda(\phi)=&(1-\nu^2)\int_{\R^3}\frac{|\sigma\cdot \nabla\phi(x)|^2}{1-V(x)+\lambda}\,dx+\nu^2\int_{\R^3}\frac{|\sigma\cdot \nabla\phi(x)|^2}{1-V(x)+\lambda}\,dx\\
&+ \int_{\R^3}V(x)|\phi(x)|^2\,dx+(1-\lambda)\int_{\R^3}|\phi(x)|^2\,dx.
\end{align*}
In~\cite{DolEstSer-00,DolEstLosVeg-04} the following Hardy-type inequality was proved
\begin{equation}
\int_{\R^3}\frac{|\sigma\cdot \nabla\phi(x)|^2}{a+1/|x|}\,dx+\int_{\R^3}\left(a-\frac{1}{|x|}\right)|\phi(x)|^2\,dx\geq0
\label{eq:Hardy-type2}
\end{equation}
for all $a>0$. Using our assumption that $V$ is bounded from below by the Coulomb potential, we can estimate 
\begin{align}
\nu^2\int_{\R^3}\frac{|\sigma\cdot \nabla\phi(x)|^2}{1-V(x)+\lambda}\,dx&\geq \nu^2\int_{\R^3}\frac{|\sigma\cdot \nabla\phi(x)|^2}{1+\nu/|x|+\lambda}\,dx\nonumber\\
&\geq \int_{\R^3}\left(\frac{\nu}{|x|}-1-\lambda\right)|\phi(x)|^2\,dx\nonumber\\
&\geq  -\int_{\R^3}V(x))|\phi(x)|^2\,dx-(1+\lambda)\int_{\R^3}|\phi(x)|^2\,dx.\label{eq:estim_kinetic_potential}
\end{align}
Thus we have proved that
\begin{equation}
q_\lambda(\phi)+2\lambda \int_{\R^3}|\phi(x)|^2\,dx\geq(1-\nu^2)\int_{\R^3}\frac{|\sigma\cdot \nabla\phi(x)|^2}{1-V(x)+\lambda}\,dx. 
\label{eq:estim_below_Q_lambda}
\end{equation}
Since the right side is positive, this shows that $q_\lambda+2\lambda\|\phi\|_{L^2}^2$ is positive as well. In addition, we conclude from~\eqref{eq:estim_kinetic_potential} that this defines a norm which is independent of $\lambda$ and is equivalent to that given by the quadratic form
\begin{equation}
\norm{\phi}_\cV^2:= \int_{\R^3}\frac{|\sigma\cdot \nabla\phi(x)|^2}{2-V(x)}\,dx+\int_{\R^3}|\phi(x)|^2\,dx.
\label{eq:norm_V}
\end{equation}
The following result provides some new properties of this space which are going to be useful for proving the min-max principle stated below in Section~\ref{sec:min-max}.

\begin{thm}[The quadratic form domain]\label{thm:quadratic form_domain}
Assume that 
$$V(x)\geq -\frac{1}{|x|}\quad\text{and}\quad \sup(V)<2$$
and let 
\begin{multline}
 \cV=\Big\{\phi\in L^2(\R^3,\C^2)\cap H^1_{\rm loc}(\R^3\setminus\{0\},\C^2)\ :\\ (2-V)^{-1/2}\sigma\cdot \nabla\phi\in L^2(\R^3,\C^2)\Big\}. 
 \label{eq:def_cV}
\end{multline}
Then $C^\ii_c(\R^3\setminus\{0\},\C^2)$ is dense in $\cV$ for the norm~\eqref{eq:norm_V}. In addition, we have the continuous embedding
$$\cV\subset H^{1/2}(\R^3,\C^2).$$
\end{thm}

Given the definition~\eqref{eq:def_cV} of the space $\cV$, the proof of Theorem~\ref{thm:quadratic form_domain} reduces to the study of a Sobolev-type space with a weight vanishing at the origin. This type of question has attracted a lot of attention and plays an important role for degenerate elliptic problems. In our proof given in Section~\ref{sec:proof_quadratic form_domain}, we follow ideas of Zhikov~\cite{Zhikov-98,Zhikov-13}.

Loosely speaking, Theorem~\ref{thm:quadratic form_domain} says that there is no ambiguity in the definition of the domain of the quadratic form $q_\lambda$. It is the same to start with the very small space $C^\ii_c(\R^3\setminus\{0\},\C^2)$ and close it for the norm $\norm{\cdot}_\cV$ (as done in~\cite{EstLos-07} for $C^\ii_c(\R^3,\C^2)$), or to directly start with the maximal domain~$\cV$ on which $q_\lambda$ is naturally defined and continuous. 

\begin{remark}
In~\eqref{eq:def_cV}, $\sigma\cdot \nabla\phi$ is understood in the sense of distributions on $\R^3$. Since $\sigma\cdot \nabla\phi\in H^{-1}(\R^3)$, it is the same to use distributional derivatives in $\R^3\setminus\{0\}$. Moreover, since $\sqrt{2-V}\in L^2_{\rm loc}(\R^3)$, we deduce from the Cauchy-Schwarz inequality that $\sigma\cdot \nabla\phi\in L^1_{\rm loc}$ for all the functions $\phi\in\cV$.
\end{remark}

Now that we have discussed the properties of the space $\cV$, we can come back to the problem of characterizing the distinguished self-adjoint extension of $D_V$. The following is a reformulation of the main result of~\cite{EstLos-07}.

\begin{thm}[$\cV$ and the distinguished extension~\cite{EstLos-07}]\label{thm:Esteban-Loss}
 Assume that for some $0\leq \nu<1$
 \begin{equation}
 V(x)\geq -\frac{\nu}{|x|}\qquad\text{and}\qquad \sup(V)<1+\sqrt{1-\nu^2}.
 \label{eq:assumption_V}
 \end{equation}
Then the distinguished self-adjoint extension $D_V$ of Theorem~\ref{thm:distinguished} is also the unique extension of the minimal operator $\dot{D}_V$ defined on $C^\ii_c(\R^3\setminus\{0\},\C^4)$, such that 
$$\cD(D_V)\subset \left\{\Psi=\begin{pmatrix}\phi\\ \chi\end{pmatrix}\in L^2(\R^3,\C^4)\ :\ \phi\in\cV\right\}.$$
More precisely, we have
$$\cD(D_V)=\left\{\Psi=\begin{pmatrix}\phi\\ \chi\end{pmatrix}\in L^2(\R^3,\C^4)\ :\ \phi\in\cV,\ D_0\Psi+V\Psi\in L^2(\R^3,\C^4)\right\},$$
where $D_0\Psi$ and $V\Psi$ are understood in the sense of distributions.
\end{thm}

This theorem was proved in~\cite{EstLos-07} using a space denoted as $\mathcal{H}_{+1}$, defined as the closure of $C^\ii_c(\R^3,\C^2)$ for the norm $\|\cdot\|_\cV$. From the density proved in our Theorem~\ref{thm:quadratic form_domain} we infer that $\mathcal{H}_{+1}=\cV$, the maximal domain on which $q_\lambda$ is continuous, and therefore Theorem~\ref{thm:Esteban-Loss} is just a reformulation of the results in~\cite{EstLos-07}.

Since only the upper component $\phi\in\cV$ appears in the statement, this characterization seems to provide less information on the domain $\cD(D_V)$. However, the following simple result says that we have $\chi\in\cV$ as well. Since $\phi,\chi\in\cV$ implies that $\phi,\chi\in H^{1/2}(\R^3,\C^2)$ by Theorem~\ref{thm:quadratic form_domain}, this means that Theorem~\ref{thm:Esteban-Loss} actually provides more information on the domain of the distinguished self-adjoint extension than Theorem~\ref{thm:distinguished}.

\begin{cor}\label{cor:chi_in_cV}
Assume that $\phi\in \cV$ and $\chi\in L^2(\R^3,\C^2)$ are such that the distribution
$$D_V\begin{pmatrix}\phi\\\chi\end{pmatrix}\ \text{belongs to}\  L^2(\R^3,\C^4),$$
where $V$ satisfies~\eqref{eq:assumption_V}.
Then $\chi\in \cV$ as well. In particular, the distinguished self-adjoint extension satisfies $\cD(D_V)\subset \cV\times\cV$.
\end{cor}

\begin{proof}
Since by assumption
$(1+V)\phi-i\sigma\cdot\nabla \chi\in L^2(\R^3,\C^2)$
and $\phi\in L^2(\R^3,\C^2)$, we also have 
$$-(2-V)\phi-i\sigma\cdot\nabla \chi\in L^2(\R^3,\C^2).$$
The function $V$ is uniformly bounded outside of the origin, hence $\chi\in H^1_{\rm loc}(\R^3\setminus\{0\},\C^2)$. Also, since $V\in L^2_{\rm loc}(\R^3)$ we have $V\phi\in L^1_{\rm loc}(\R^3,\C^2)$. Therefore $\sigma\cdot\nabla \chi\in L^1_{\rm loc}$ as well. Using that $(2-V)^{-1/2}$ is bounded, we deduce that 
$$-(2-V)^{1/2}\phi-(2-V)^{-1/2}i\sigma\cdot\nabla \chi\in L^2(\R^3,\C^2).$$
From~\eqref{eq:estim_kinetic_potential} we know that $(2-V)^{1/2}\phi\in L^2(\R^3,\C^2)$ hence conclude, as we wanted, that $(2-V)^{-1/2}\sigma\cdot\nabla \chi\in L^2(\R^3,\C^2)$.
\end{proof}

\subsection{The critical case $\nu=1$}\label{sec:critical}
We give in this section some new properties of the distinguished self-adjoint extension in the critical case. Although these will not all be needed for the min-max formulas in Section~\ref{sec:min-max}, we state them because they complement~\cite{EstLos-07,EstLos-08} in an interesting direction. 

The Esteban-Loss method presented in the previous section is general and it was applied to the critical case already in~\cite{EstLos-07}. The main difficulty here is to understand the domain of $q_\lambda$, since the inequality~\eqref{eq:estim_below_Q_lambda} does not give any useful information when $\nu=1$. The terms in $q_\lambda$ will not necessarily be separately finite. Following ideas from~\cite{DolEstLosVeg-04,DolEstDuoVeg-07}, we first describe this domain with more details.

It is useful to start with the exact Coulomb case $V_{\rm C}(x)=-|x|^{-1}$, in which case we use the notation
\begin{equation}
q_\lambda^{\rm C}(\phi)=\int_{\R^3}\left\{\frac{|x|}{1+\lambda|x|+|x|}\left|\sigma\cdot\nabla\phi(x)\right|^2+\left(1-\lambda-\frac{1}{|x|}\right)|\phi(x)|^2\right\}\,dx.
\label{eq:def_q_lambda_Coulomb}
\end{equation}
Our aim is to understand what is the \emph{maximal} domain on which $q_\lambda^{\rm C}$ is well-defined and continuous. To this end, we start with $\lambda=0$ and follow~\cite{DolEstLosVeg-04}. We involve the operator $k=1+\sigma\cdot L$, where 
$$L=-ix\wedge \nabla=-i\begin{pmatrix}
x_2\partial_3-x_3\partial_2\\
x_3\partial_1-x_1\partial_3\\
x_1\partial_2-x_2\partial_1
\end{pmatrix}$$
is the angular momentum. We recall that $k=1+\sigma\cdot L$ has the eigenvalues $\pm1,\pm2,...$, see~\cite{Thaller}. The negative and positive spaces are unitarily equivalent and mapped to one another using the unitary $\sigma\cdot \omega_x$ where $\omega_x=x/|x|$ is the unit vector pointing in the same direction as $x$:
\begin{equation}
\sigma\cdot\frac{x}{|x|}\big(1+\sigma\cdot L\big)\sigma\cdot\frac{x}{|x|}=-\big(1+\sigma\cdot L\big).
\label{eq:negative_positive_unitary}
\end{equation}
In addition, we will use that the kernel of $\sigma\cdot L$ is composed of radial functions (it coincides with the kernel of $L$), hence the kernel of $\sigma\cdot L+2$ is given by $\sigma\cdot \omega_x$ times radial functions. These are the two spaces for the upper spinor $\phi$ which correspond to $\kappa=\pm1$ for the full Dirac operator. The sectors $\kappa=\pm1$  determine the possible extensions, as we have recalled in Subsection~\ref{sec:radial}.
The following is inspired by~\cite{DolEstLosVeg-04,DolEstDuoVeg-07} and proved in Appendix~\ref{sec:proof_thm_q_C} below.

\begin{thm}[Writing $q_\lambda^{\rm C}$ as a sum of squares]\label{thm:carres_q_C}
For every $\phi\in L^2(\R^3,\C^2)$ we write
$$\phi=\phi_+(x)+\phi_-(x)+\phi_0(|x|)+\sigma\cdot\frac{x}{|x|}\,\phi_1(|x|)$$
where $\phi_+=\1_{[1,\ii)}(\sigma\cdot L)\phi$, $\phi_-=\1_{(-\ii,-3]}(\sigma\cdot L)\phi$, $\phi_0=\1_{\{0\}}(\sigma\cdot L)\phi$ and $\phi_1=\sigma\cdot (x/|x|)\1_{\{-2\}}(\sigma\cdot L)\phi$. Then
\begin{align}
q_0^{\rm C}(\phi)=&\int_{\R^3}\frac{|x|}{1+|x|}\left|\sigma\cdot\nabla\phi_+(x)+\frac{\sigma\cdot x}{|x|^2}(1+|x|)\,\phi_+(x)\right|^2\,dx\nn\\
&+\int_{\R^3}\frac{|x|}{1+|x|}\left|\sigma\cdot\nabla\phi_-(x)-\frac{\sigma\cdot x}{|x|^2}(1+|x|)\,\phi_-(x)\right|^2\,dx\nn\\
&+2\pscal{\phi_+,\frac{\sigma\cdot L}{|x|}\phi_+}+2\pscal{\phi_-,\frac{-2-\sigma\cdot L}{|x|}\phi_-}\nn\\
&+4\pi\int_0^\ii\frac{r}{1+r}\left|r\phi_0'(r)+\phi_0(r)+r\phi_0(r)\right|^2\,dr\nn\\
&+4\pi\int_0^\ii\frac{r}{1+r}\left|r\phi_1'(r)+\phi_1(r)-r\phi_1(r)\right|^2\,dr\label{eq:carres_q_0}
\end{align}
for every $\phi\in H^1(\R^3,\C^2)$. On $L^2(\R^3,\C^2)$, the quadratic form $q_0^{\rm C}$ is equivalent to
\begin{align}\label{eq:norm_q_0_simplifiee}
&\|\phi\|_{L^2}^2+q_0^{\rm C}(\phi)\nn\\
&\qquad\sim \|\phi\|_{L^2}^2+\int_{\R^3}\frac{|x|}{1+|x|}\big|\sigma\cdot\nabla\phi_+(x)\big|^2\,dx+\int_{\R^3}\frac{|x|}{1+|x|}\big|\sigma\cdot\nabla\phi_-(x)\big|^2\,dx\nn\\
&\qquad\qquad+\int_{0}^\ii\frac{r}{1+r}\big|r\phi_0'(r)+\phi_0(r)\big|^2dr+\int_{0}^\ii\frac{r}{1+r}\big|r\phi_1'(r)+\phi_1(r)\big|^2\,dr\nn\\
&\qquad\sim \|\phi\|_{L^2}^2+\int_{\R^3}\frac{1}{|x|(1+|x|)}\big|\sigma\cdot\nabla|x|\phi(x)\big|^2\,dx.
\end{align}
Finally, for all $-1<\lambda<1$, we have 
\begin{equation}
q_\lambda^{\rm C}(\phi)= q_0^{\rm C}(\phi)-\lambda\int_{\R^3}\frac{|x|^2|\sigma\cdot\nabla\phi(x)|^2}{(1+|x|)(1+(1+\lambda)|x|)}\,dx-\lambda\int_{\R^3}|\phi(x)|^2\,dx
\label{eq:relation_q_lambda_0_Coulomb}
\end{equation}
which, in $L^2$,  is equivalent to the norm associated with $\|\phi\|_{L^2}^2+q_0^{\rm C}(\phi)$.
\end{thm}

Note that all the terms in the formula~\eqref{eq:carres_q_0} for $q_0^{\rm C}$ are non-negative, which enables us to identify its maximal domain. We see that the two functions $\phi_+$ and $\phi_-$ have the exact same regularity as before, namely they must belong to the space $\cV_{\rm C}$, defined as in~\eqref{eq:def_cV} with $V(x)=V_{\rm C}(x)=-|x|^{-1}$:
$$\int_{\R^3}\frac{|x|}{1+|x|}\big|\sigma\cdot\nabla\phi_\pm(x)\big|^2\,dx<\ii.$$
In particular, from Theorem~\ref{thm:quadratic form_domain} and the Hardy-type inequality~\eqref{eq:Hardy-type2}, $\phi_+$ and $\phi_-$ have a finite Coulomb energy and a finite $H^{1/2}$ norm. 
Only the functions $\phi_0$ and $\sigma\cdot \omega_x\phi_1$ can be more singular at the origin. Those only satisfy the property that 
$$\int_{\R^3}\frac{1}{|x|(1+|x|)}\big|\sigma\cdot\nabla|x|\phi_{0}(x)\big|^2\,dx<\ii$$
which can be written in radial coordinates as 
\begin{equation}
\int_0^\ii\frac{r}{1+r}\left|r\phi_{0/1}'(r)+\phi_{0/1}(r)\right|^2\,dr<\ii.
\label{eq:condition_radial_critique}
\end{equation}
This is weaker than when $|x|$ is pulled outside of the gradient as before. For instance, the ground state of the Dirac-Coulomb operator at $\nu=1$ is given by~\cite[Sec.~7.4.2]{Thaller}
$$\phi_0(|x|)=\frac{e^{-|x|}}{|x|}v,\qquad v\in\C^2$$
and it satisfies $q_0^{\rm C}(\phi_0)<\ii$ but
$$\int_{\R^3}\frac{|x|}{1+|x|}\left|\sigma\cdot\nabla\frac{e^{-|x|}}{|x|}\right|^2\,dx=4\pi\int_0^\ii \frac{r^3}{1+r}\left|\frac{e^{-r}+re^{-r}}{r^2}\right|^2\,dr=+\ii.$$
The condition~\eqref{eq:condition_radial_critique} is enough to distinguish a self-adjoint extension, as we will see. The main message is that $\phi_{0}$ and $\phi_1$ are allowed to behave like $1/r$ at $r=0$, but not like $\log(r)/r$. This corresponds to taking $\alpha=\pi$ in Subsection~\ref{sec:radial}.

Now we are able to define the spaces which will replace $\cV$ in the critical case. In the exact Coulomb case $V(x)=-|x|^{-1}$ we introduce
\begin{equation}
\cW_{\rm C}=\bigg\{\phi\in L^2(\R^3,\C^2)\ :\ \frac{\sigma\cdot\nabla |x|\phi}{|x|^{1/2}(1+|x|)^{1/2}}\in L^2(\R^3,\C^2)\bigg\}.
\end{equation}
Then we assume that $V(x)\geq -|x|^{-1}$ and that $\sup(V)<1$. The quadratic form associated with $V$ can be written in terms of $q_\lambda^{\rm C}$ as follows:
\begin{multline*}
q_{\lambda}(\phi)=\int_{\R^3}\left(\frac{1}{1+\lambda-V(x)}-\frac{|x|}{1+(1+\lambda)|x|}\right)\left|\sigma\cdot\nabla\phi(x)\right|^2\,dx\\
+\int_{\R^3}\left(V(x)+\frac1{|x|}\right)|\phi(x)|^2\,dx+q_\lambda^{\rm C}(\phi),
\end{multline*}
for $-1+\sup(V)<\lambda<1$. The quadratic forms 
$$\int_{\R^3}\left(\frac{1}{1+\lambda-V(x)}-\frac{|x|}{1+(1+\lambda)|x|}\right)\left|\sigma\cdot\nabla\phi(x)\right|^2\,dx$$
are all equivalent when $\lambda$ is varied in the interval $(-1+\sup(V),1)$ and since the same holds for $q_\lambda^{\rm C}$ by Theorem~\ref{thm:carres_q_C}, we can simply use $\lambda=0$ and define the space $\cW$ associated with $V$ by
\begin{multline}
\cW=\bigg\{\phi\in \cW_{\rm C} \ :\ \left(\frac{1}{1-V(x)}-\frac{|x|}{1+|x|}\right)^{1/2}\sigma\cdot \nabla\phi\in L^2(\R^3,\C^2),\\
\left(V(x)+\frac{1}{|x|}\right)^{1/2}\phi\in L^2(\R^3,\C^2)\bigg\}.
\end{multline}

The following is the equivalent of Theorem~\ref{thm:quadratic form_domain}.

\begin{thm}[Properties of $\cW_{\rm C}$ and $\cW$]\label{thm:quadratic form_domain_critical}
We assume that 
 \begin{equation}
 V(x)\geq -\frac{1}{|x|}\qquad\text{and}\qquad \sup(V)<1.
 \label{eq:assumption_V_critical}
 \end{equation}
Then the space $C^\ii_c(\R^3\setminus\{0\},\C^2)$ is dense in $\cW_{\rm C}$ and in $\cW$ for their respective norms. Also, we have the continuous embeddings 
$$\cW\subset \cW_{\rm C}\subset H^s(\R^3,\C^2)$$
for every $0\leq s<1/2$.
\end{thm}

The proof of Theorem~\ref{thm:quadratic form_domain_critical} is provided below in Section~\ref{sec:proof_thm_quadratic form_domain_critical} and it is much more involved than that of Theorem~\ref{thm:quadratic form_domain}. This is due to the criticality of the problem, which prevents from using rough regularization techniques.

\begin{remark}
If $V(x)=-|x|^{-1} +O\left(|x|^{-\alpha}\right)$ with $\alpha<1$, as $x\to0$, then we simply have $\cW=\cW_{\rm C}$. Indeed the two additional terms are controlled by the $\cW_{\rm C}$-norm. We have
\begin{align*}
&\int_{\R^3}\left(\frac{1}{1-V(x)}-\frac{|x|}{1+|x|}\right)\left|\sigma\cdot\nabla\phi(x)\right|^2\,dx\\
&\qquad\qquad\lesssim \int_{\R^3}\frac{|x|^2}{|x|^\alpha(1+|x|)^2}\left|\sigma\cdot\nabla\phi(x)\right|^2\,dx\\
&\qquad\qquad\lesssim \int_{\R^3}\frac{1}{|x|^\alpha(1+|x|)^2}\left|\sigma\cdot\nabla|x|\phi(x)\right|^2\,dx+\int_{\R^3}\frac{|\phi(x)|^2}{|x|^\alpha(1+|x|)^2}\,dx
\end{align*}
and, similarly,
$$\int_{\R^3}\left(V(x)+\frac1{|x|}\right)|\phi(x)|^2\,dx\lesssim \int_{\R^3}\frac{|\phi(x)|^2}{|x|^\alpha}\,dx$$
which are all finite for $\phi\in \cW_{\rm C}$. Hence in this case there is no difference between $\cW$ and $\cW_{\rm C}$.
\end{remark}

Contrary to the subcritical case where one can use the space $H^{1/2}$, we cannot distinguish the extension from the sole property that it is included in $H^s$ for $s<1/2$. This would not make the difference between $1/r$ and $\log(r)/r$. We need the more precise norm associated with $q_0$. The main result on the distinguished self-adjoint extension is the following.

\begin{thm}[$\cW$ and the distinguished extension in the critical case]\label{thm:distinguished-critical}
 We assume that 
 \begin{equation}
 V(x)\geq -\frac{1}{|x|}\qquad\text{and}\qquad \sup(V)<1.
 \label{eq:assumption_V_critical_bis}
 \end{equation}

 \smallskip
 
\noindent {\rm (a)} \cite{EstLos-07} The minimal operator $\dot{D}_V=(D_0+V)\upharpoonright C^\ii_c(\R^3\setminus\{0\},\C^4)$ has a unique self-adjoint extension $D_V$ such that 
$$\cD(D_V)\subset \left\{\Psi=\begin{pmatrix}\phi\\ \chi\end{pmatrix}\in L^2(\R^3,\C^4)\ :\ \phi\in\cW\right\}$$
and this extension has the domain
$$\cD(D_V)=\left\{\Psi=\begin{pmatrix}\phi\\ \chi\end{pmatrix}\in L^2(\R^3,\C^4)\ :\ \phi\in\cW,\ D_0\Psi+V\Psi\in L^2(\R^3,\C^4)\right\}$$
where $D_0\Psi$ and $V\Psi$ are understood in the sense of distributions. 

 \smallskip

 \noindent {\rm (b)} Let $V_\eps(x):=\max(V(x),-1/\eps)$ or $V_\eps=(1-\eps)V$. Then, the self-adjoint operator $D_{V_\eps}$ converges in the norm resolvent sense to the operator $D_V$ defined in the previous item. 
\end{thm}

Although the first part is just a reformulation of the results in~\cite{EstLos-07} (relying on the closure $\mathcal{H}_{+1}$ of $C^\ii_c(\R^3,\C^2)$ for the norm induced by $q_0$, which is equal to $\cW$ by Theorem~\ref{thm:quadratic form_domain_critical}), the convergence of the resolvents is completely new. In the same spirit as what was achieved for $\nu<1$ in~\cite{Wust-73,Wust-75,Wust-77,KlaWus-79,Kato-83}, it means that the Esteban-Loss extension is the only physically relevant one in the critical case. The proof of the resolvent convergence is given in Section~\ref{sec:proof_thm_distinguished_critical} below.

\section{Domains for min-max formulas of eigenvalues}\label{sec:min-max}

In this section we finally discuss min-max principles for Dirac eigenvalues. In \cite{DolEstSer-00} an abstract variational characterization of
the eigenvalues of operators with gaps was shown. Let ${\gH}$ be a Hilbert space and $A:D(A) \subset {\gH}\rightarrow {\gH}$ be a self-adjoint operator. Let
${\gH}^+$, ${\gH}^-$ be two orthogonal Hilbert subspaces of 
${\gH}$ such that ${\gH}={\gH}^+\oplus {\gH}^-$.
We denote by $\Lambda^\pm$ the two corresponding orthogonal projectors. We assume the
existence of a core $F$ (a subspace of $D(A)$ which is dense for the
norm $\|\cdot\|_{D(A)}$), such that
\begin{itemize}
\item[$(i)$] $F^+ = \Lambda^+ F$ and $F^- = \Lambda^- F$ are two subspaces
of $\cD(|A|^{1/2})$,
\item[$(ii)$] $\dps a=\sup_{\psi^- \in F^-\setminus \{ 0\}} \frac{\pscal{\psi^-, A\psi^-}_{\gH}}{\|\psi_-\|_\gH^{2}}<+\infty $ .\medskip
\end{itemize}
We then consider the sequence of min-max levels
\bq \lambda_F^{(k)} := \  \inf_{
 \scriptstyle W \ {\rm subspace \ of \ } F^+  \atop  \scriptstyle {\rm dim}
\ W =
k  } \  \Sup_{  \scriptstyle \psi \in ( W \oplus F^- ) \setminus \{ 0 \} } \
\Frac{\pscal{\psi, A\psi}_\gH}{\|\psi\|^2_{_{\gH}}} \ ,
\qquad k \geq 1.
\label{min-max} \eq
Our last assumption is
\begin{itemize}
\item[$(iii)$] $\qquad\qquad \lambda_F^{(1)} > a$.
\end{itemize}

Everywhere $\pscal{\psi,A\psi}=\pscal{|A|^{1/2}\psi,U|A|^{1/2}\psi}$ is always understood in the form sense, which is possible since $F^\pm\subset\cD(|A|^{1/2})$. Let 
$$b = \inf \ (\sigma_{\rm ess} (A) \cap (a, + \infty)) \in [a, +
\infty]$$
be the bottom of the essential spectrum above $a$. The following gives a characterization of the eigenvalues in the gap $(a,b)$.

\begin{thm}[Min-max formula for the $k$th eigenvalue~\cite{DolEstSer-00}]\label{thm:DES}  With the above notations, and under
assumptions $(i)$---$(iii)$, we have $b>a$. The number $ \lambda_F^{(k)}$  is the $k$th eigenvalue of $A$ in $(a,b)$, counted with multiplicity, or is equal to $b$ if $A$ has less than $k$ eigenvalues in $(a,b)$. 
\end{thm}

For the Dirac operator, it was suggested by Talman~\cite{Talman-86} and Datta-Devaiah~\cite{DatDev-88} to use the decomposition into upper and lower spinors, that is, to take for the two subspaces $\gH^\pm$
\begin{equation}
 \gH^+=\left\{\begin{pmatrix}\phi\\0\end{pmatrix}\ :\ \phi\in L^2(\R^3,\C^2)\right\},\quad \gH^-=\left\{\begin{pmatrix}0\\ \chi\end{pmatrix}\ :\ \chi\in L^2(\R^3,\C^2)\right\}. 
 \label{eq:Talman}
\end{equation}
The first rigorous result for this decomposition was obtained by Griesemer and Siedentop~\cite{GriSie-99}, who dealt with bounded potentials $V$. In~\cite{DolEstSer-00} the above abstract result was applied to the case of Coulomb singularities. However,in \cite{DolEstSer-00} it was stated that it is possible to use the space $F=C^\infty_c(\R^3, \C^4)$. From \textit{(4)} in Theorem~\ref{thm:distinguished}, this is true when $0\leq \nu<{\sqrt{3}}/{2}$, because in this range the operator $D_\nu$ is essentially self-adjoint on $C^\infty_c(\R^3\setminus\{0\}, \C^4)$. When $\sqrt{3}/2\leq \nu<1$, the argument in~\cite{DolEstSer-00} was not complete. 

Of course, Theorem~\ref{thm:DES} can still be applied in the domain $\cD(D_V)$ of the distinguished self-adjoint extension or in any core $F$ on which $D_V$ is essentially self-adjoint. Unfortunately, except for bounded perturbations of the exact Coulomb potential, for which the domain is well understood as we have seen in Subsection~\ref{sec:radial}, $\cD(D_V)$ is not so easy to grasp for a general potential $V$. From a numerical point of view, it is indeed important to be able to use simple spaces $F$ in the min-max formula. 

In~\cite{MorMul-15,Muller-16}, Müller and Morozov proved the validity of the min-max formula for $\sqrt{3}/2\leq \nu<1$ in $F=H^{1/2}(\R^3,\C^2)$, using a variant of the abstract min-max theorem in a setting adapted to form domains, inspired by Nenciu~\cite{Nenciu-76}.

Another min-max principle based on the free-energy projectors $\Lambda_0^+=\1(D_0\geq0)$ and $\Lambda_0^-=\1(D_0\leq0)$ was first introduced in \cite{EstSer-97}. Using an inequality proved in \cite{BurEva-98} and \cite{Tix-98}, it was shown in \cite{DolEstSer-00} that the eigenvalues satisfy the min-max principle~\eqref{eq:min-max_principle} in the range $\,0\leq\nu<2\left(\frac{\pi}{2}+\frac{2}{\pi}\right)^{-1}\simeq 0.9$. 
Recently, the free projections have also been covered in~\cite{MorMul-15,Muller-16} for $\nu<1$. 

In this section we prove a result similar to~\cite{MorMul-15,Muller-16}, by a completely different method. We will show that the min-max is valid on any space between $C^\infty_c(\R^3\setminus\{0\}, \C^4)$ and $H^{1/2}(\R^3,\C^4)$. Contrary to~\cite{MorMul-15} we will not modify the abstract theorem, but simply use density results in the spirit of Theorem~\ref{thm:quadratic form_domain}. We will also treat the critical case $\nu=1$ and obtain the first results in this setting, to our knowledge. 

In order to properly state our main result, we introduce the two projections 
$$\Lambda_T^+\begin{pmatrix}\phi\\ \chi\end{pmatrix}=\begin{pmatrix}\phi\\0\end{pmatrix},\qquad \Lambda_T^-\begin{pmatrix}\phi\\ \chi\end{pmatrix}=\begin{pmatrix}0\\\chi\end{pmatrix}$$
corresponding to the Talman decomposition~\eqref{eq:Talman} and the spectral projections
$$\Lambda_0^+=\1(D_0\geq 0),\qquad \Lambda_0^-=\1(D_0\leq 0)$$
of the free Dirac operator. For a space $F\subseteq H^{1/2}(\R^3,\C^4)$, we define the min-max levels
\begin{equation}
 \lambda_{T/0,F}^{(k)}= 
 \  \inf_{\substack{W \ {\rm subspace \ of \ } \Lambda_{T/0}^+F \\ \scriptstyle {\rm dim}
\ W=k  }} \  \Sup_{\substack{\Psi \in W \oplus \Lambda_{T/0}^-F \\ \Psi\neq0}} \
\Frac{\pscal{\Psi, D_V\Psi}}{\|\Psi\|^2_{_{L^2}}},
\qquad k \geq 1.
 \label{eq:min-max_principle}
\end{equation}
We remark that the four projections $\Lambda_{T/0}^\pm$ stabilize $H^{1/2}(\R^3,\C^4)$, hence $\pscal{\Psi, D_V\Psi}$ is always well defined in the sense of quadratic forms. Indeed
$$\pscal{\Psi, D_V\Psi}=\pscal{\Psi, D_0\Psi}+\int_{\R^3}V|\Psi|^2$$ 
by Theorem~\ref{thm:distinguished} {\textit{(ii)} and Remark~\ref{rmk:form}. The same property as in Remark~\ref{rmk:form} holds in the critical case $\nu=1$, since $H^1(\R^3,\C^4)\subset\cD(D_V)$ as well.
We could actually work in $\cD(|D_V|^{1/2})$ but we refrain from doing it since our goal is to state a result in simple spaces that do not depend on $V$.

Our main result is the following
\begin{thm}[Min-max formula for eigenvalues]\label{thm:new}  
Let $0<\nu\leq 1$. We assume that
 \begin{equation}
 V(x)\geq -\frac{\nu}{|x|}\qquad\text{and}\qquad \sup(V)<1+\sqrt{1-\nu^2}.
 \label{eq:assumption_V_bis}
 \end{equation}
Let
\begin{equation}
 C^\ii_c(\R^3\setminus\{0\},\C^4)\;\subseteq\; F\; \subseteq\; H^{1/2}(\R^3,\C^4). 
 \label{eq:condition_F}
\end{equation}
Then, the number $\lambda_{T,F}^{(k)}$ defined in \eqref{eq:min-max_principle}, is independent of the subspace $F$ and coincides with the $k$th eigenvalue of the distinguished self-adjoint extension of $D_V$ larger than or equal to $\sqrt{1-\nu^2}$, counted with multiplicity (or is equal to $b = \inf \; (\sigma_{\rm ess} (D_V) \cap (\sqrt{1-\nu^2}, + \infty))$ if there are less than $k$ eigenvalues below $b$). In addition, we have
$$\lambda_{T,F}^{(k)}=\lambda_{0,F}^{(k)}$$
for all $F$ as above and all $k\geq1$.
\end{thm}

That we can take any space $F$ satisfying~\eqref{eq:condition_F} shows how the min-max characterization of the eigenvalues is insensitive to $F$, even for the distinguished self-adjoint extension which has a non trivial domain $\cD(D_V)$. The space $F$ can be as small as $C^\ii_c(\R^3\setminus\{0\},\C^4)$ which is not dense in $\cD(D_V)$ for $\sqrt{3}/2<\nu\leq1$, or as large as $H^{1/2}(\R^3,\C^4)$ which does not even contain the domain for $\nu=1$. 

Before turning to the proof of the theorem (given in Section~\ref{sec:proof_thm_new}), we would like to comment on the role of the quadratic form $q_\lambda$ discussed in Sections~\ref{sec:Esteban-Loss}--\ref{sec:critical}, in the Talman case $\Lambda^\pm_T$. One important argument in~\cite{DolEstSer-00} was to solve the sup part of~\eqref{eq:min-max_principle} using the method of Lagrange multipliers. For any $\lambda>\sup(V)-1$ we consider the maximization problem
\begin{multline*}
\sup_{{\scriptsize \begin{pmatrix}0\\ \chi\end{pmatrix}}\in \Lambda_T^-F}\left\{\pscal{\begin{pmatrix}\phi\\ \chi\end{pmatrix},D_V\begin{pmatrix}\phi\\ \chi\end{pmatrix}}-\lambda \big(\norm{\phi}_{L^2}^2+\norm{\chi}_{L^2}^2\big) \right\}\\
=\int_{\R^3}\frac{|\sigma\cdot\nabla\phi(x)|^2}{1-V(x)+\lambda}\,dx+\int_{\R^3}\big(1+V(x)-\lambda\big)|\phi(x)|^2\,dx=q_\lambda(\phi),
\end{multline*}
which is exactly the quadratic form which we have studied in Section~\ref{sec:Esteban-Loss}. The unique maximizer is
$$\chi=\frac{-i\sigma\cdot\nabla\phi}{1-V+\lambda}.$$
This can be used to prove that supremum
\begin{equation}
\sup_{{\scriptsize \begin{pmatrix}0\\ \chi\end{pmatrix}}\in \Lambda_T^-F}\frac{\pscal{\begin{pmatrix}\phi\\ \chi\end{pmatrix},D_V\begin{pmatrix}\phi\\ \chi\end{pmatrix}}}{\norm{\phi}^2+\norm{\chi}^2}
\end{equation}
appearing in the min-max formula~\eqref{eq:min-max_principle}, is the unique number $\lambda$ such that $q_\lambda(\phi)=0$. 
For this reason, our proof of Theorem~\ref{thm:new} relies on the density of $C^\ii_c(\R^3\setminus\{0\},\C^2)$ in the quadratic form domains $\cV$, shown in Theorem~\ref{thm:quadratic form_domain}. In the critical case, our proof does not rely on the density in $\cW$, stated before in Theorem~\ref{thm:quadratic form_domain_critical}. This is because we have assumed that $F\subset H^{1/2}(\R^3,\C^4)$ and $\cW\cap H^{1/2}(\R^3,\C^4)=\cV$. 

The rest of the paper is dedicated to the proofs of our results.


\section{Proof of Theorem~\ref{thm:new} on the min-max levels}\label{sec:proof_thm_new}

Admitting temporarily our other results, we start with the proof of Theorem~\ref{thm:new}. One possible route is to apply the abstract Theorem~\ref{thm:DES} in the domain $F_0=\cD(D_V)$ and then to show that $F_0$ can be replaced by any other $F$ as in the statement. Another strategy is to truncate the potential into $V_\eps$, apply Theorem~\ref{thm:DES} for $V_\eps$ and then pass to the limit $\eps\to0$ in the min-max formula for the eigenvalues. This argument uses the norm-convergence of the resolvent in Theorems~\ref{thm:distinguished} and~\ref{thm:distinguished-critical} which implies the convergence of the eigenvalues. 

The first method requires to know the domain $F_0=\cD(D_V)$ quite precisely, whereas the second one does not involve the domain at all. It is more robust and more appropriate in the critical case $\nu=1$ for which we have less information on $\cD(D_V)$. For this reason, we use the second method.

\subsection{Proof for the Talman projections $\Lambda^\pm_T$} 

We split the proof into several steps. To simplify our proof, with an abuse of notation we write $\phi\in F^+=\Lambda_T^+ F$ instead of
$$\begin{pmatrix}\phi\\ 0\end{pmatrix}\in \Lambda_T^+ F$$
and similarly we write $\chi\in F^-$. In the proof we approximate the (upper bounded) potential $V$ by $V_\eps:=\max(V,-1/\eps)\in L^\ii(\R^3,\R)$ and we start by recalling some well-known facts for $V_\eps$. 

\subsubsection*{Step 1. Upper bound}
First we compute 
\begin{equation*}
a:=\sup_{\substack{\chi\in F^-\\ \chi\neq0}} \frac{\pscal{\begin{pmatrix}0\\ \chi\end{pmatrix},D_V\begin{pmatrix}0\\ \chi\end{pmatrix}}}{\|\chi\|^2}
=\sup_{\substack{\chi\in F^-\\ \chi\neq0}}\frac{\int_{\R^3}(-1+V)|\chi|^2}{\|\chi\|^2}=\sup(-1+V)
\end{equation*}
since $F^-$ contains $C^\ii_c(\R^3\setminus\{0\},\C^2)$ by assumption. Thus $a<\sqrt{1-\nu^2}$ since $\sup(V)<1+\sqrt{1-\nu^2}$. The same property holds when $V$ is replaced by $V_\eps$.

Following~\cite[Lem.~2.2]{DolEstSer-00}, we write the min-max levels for a potential $V$ (truncated or not) in the form
\begin{equation}
\lambda^{(k)}_{T,F}(V)=\inf_{\substack{W^+\subset F^+\\ \dim(W^+)=k}}\sup_{\substack{\phi\in W^+}}S_{F^-}(V,\phi).
\label{eq:lambda_k_S_critical}
\end{equation}
where
\begin{multline}
S(V,\phi):=\\ \sup_{\substack{\chi\in F^-\\ \|\phi\|^2+\|\chi\|^2\neq0}}
\frac{\int_{\R^3}(|\phi|^2-|\chi|^2)+\int_{\R^3}V(|\phi|^2+|\chi|^2)+2\Re \pscal{\chi,-i\sigma\cdot\nabla\phi}}{\int_{\R^3}|\phi|^2+|\chi|^2}. 
\label{eq:def_S}
\end{multline}
All the terms are well defined since $F\subset H^{1/2}(\R^3,\C^4)$. Indeed, by continuity of the function appearing in the definition~\eqref{eq:def_S} for the norm of $H^{1/2}$, the value of $S(V,\phi)$ does not depend on $F^-$ which can be replaced by any space dense in $H^{1/2}$. This is why our notation for $S(V,\phi)$ does not involve $F^-$. 
By monotonicity with respect to $V$ we have
\begin{equation}
 \lambda^{(k)}_{T,F}(V_\eps)\geq \lambda^{(k)}_{T,F}(V) 
 \label{eq:upper_bound}
\end{equation}
for all $\eps>0$. 
Using a continuation principle, it was proved in~\cite{DolEstSer-00} that 
\begin{equation}
 \lambda_{T,F}^{(1)}(V_\eps)\geq\sqrt{1-\nu^2}>a
 \label{eq:lower_bound_lambda_1}
\end{equation}
for all $\eps>0$ and all $C^\ii_c(\R^3\setminus\{0\},\C^4)\subset F\subset H^{1/2}(\R^3,\C^4)$. 
So we can apply Theorem~\ref{thm:DES} and conclude that, under our assumptions on $V$, $\lambda^{(k)}_{T,F}(V_\eps)$ is independent of $F$ and coincides with the $k$th eigenvalue of $D_{V_\eps}$. 

In the limit $\eps\to0$, $\lambda^{(k)}_{T,F}(V_\eps)$ converges to the $k$th eigenvalue $\mu^{(k)}(V)$ of $D_V$, due to the convergence in norm of the resolvents, shown in Theorem~\ref{thm:distinguished} for the subcritical case $0<\nu<1$ and in Theorem~\ref{thm:distinguished-critical} in the critical case $\nu=1$. So passing to the limit $\eps\to0$ in~\eqref{eq:upper_bound} we obtain the upper bound
$$\mu^{(k)}(V)\geq \lambda^{(k)}_{T,F}(V).$$

\subsubsection*{Step 2. Lower bound}
Now we come back to~\eqref{eq:lambda_k_S_critical}. In order to prove the reverse inequality we have to show that
$$\sup_{\substack{\phi\in W^+}}S(V,\phi)\geq \mu^{(k)}(V)$$
for every $k$-dimensional subspace $W^+\subset F^+$. The next lemma follows from the arguments in~\cite[Lemma~2.2]{DolEstSer-00}. 

\begin{lemma}[Computation of $S(V,\phi)$~\cite{DolEstSer-00}]\label{lem:value_S}
Let $\phi\in H^{1/2}(\R^3,\C^2)$. Then $S(V,\phi)$ is finite if and only if $\phi\in\cV$. In this case, $E=S(V,\phi)$ is the unique solution to the nonlinear equation $q_E(V,\phi)=0$.
\end{lemma}

Note that even in the critical case $\nu=1$ we conclude that $\phi$ must be in $\cV$. This is because we have assumed that $\phi\in H^{1/2}(\R^3,\C^2)$ and $\cW\cap H^{1/2}=\cV$. 

By Lemma~\ref{lem:value_S} and the monotonicity of $q_E$ with respect to $E$, it suffices to show that
\begin{equation}
\sup_{\phi\in W^+} q_{\mu^{(k)}(V)}(V,\phi)\geq 0
\label{eq:to_be_shown_critical_q}
\end{equation}
for any $k$-dimensional space $W^+\subset F^+\cap \cV$. Since $C^\ii_c(\R^3\setminus\{0\},\C^2)$ is dense in $\cV$ by Theorem~\ref{thm:quadratic form_domain}, it suffices to prove~\eqref{eq:to_be_shown_critical_q} for a $k$-dimensional space $W^+\subset C^\ii_c(\R^3\setminus\{0\},\C^2)$. For any such space, we have from the min-max characterization for $V_\eps$
$$\sup_{\phi\in W^+} q_{\mu^{(k)}(V)}(V_\eps,\phi)\geq \sup_{\phi\in W^+} q_{\mu^{(k)}(V_\eps)}(V_\eps,\phi)\geq 0.$$
So passing to the limit $\eps\to0$ (in the fixed finite-dimensional space $W^+\subset C^\ii_c(\R^3\setminus\{0\},\C^2)$) we find 
$$\sup_{\phi\in W^+} q_{\mu^{(k)}(V)}(V,\phi)\geq 0$$
as we wanted. 

\subsection{Proof for the free Dirac projections $\Lambda^\pm_0$} 

The proof follows along the same lines as for the Talman projections, and we only outline it. In this case we have as before
\begin{align*}
a:=\sup_{\substack{\psi_-\in\Lambda_0^-F\\\psi_-\neq0}}\frac{\pscal{\psi_-,D_V\psi_-}}{\|\psi_-\|^2}
&=\sup_{\substack{\psi_-\in\Lambda_0^-F\\\psi_-\neq0}}\frac{\pscal{\psi_-,\big(-\sqrt{1-\Delta}+V\big)\psi_-}}{\|\psi_-\|^2}\\
&\leq \sup(V)-1<\sqrt{1-\nu^2}.
\end{align*}
Following step by step the argument of the previous section, we have to study the quadratic form
\begin{multline}\label{eq:q_tilde_E}
\tilde q_E(\psi_+):=\pscal{\psi_+,\sqrt{1-\Delta}\,\psi_+}+\int_{\R^3}(V-E)|\psi_+|^2\\
+\pscal{\Lambda_0^-V\psi_+, \left(\Lambda_0^-(\sqrt{1-\Delta}+E-V)\Lambda_0^-\right)^{-1}\Lambda_0^-V\psi_+}
\end{multline}
in place of $q_E$ which appeared in \eqref{eq:q_E}. 
Let us remark that $\tilde q_E$ is continuous on $H^{1/2}$ since, by the operator monotonicity of the inverse, we have
$$\left(\Lambda_0^-(\sqrt{1-\Delta}+E-V)\Lambda_0^-\right)^{-1}\leq \frac{\Lambda_0^-}{\sqrt{1-\Delta}+E-1-\sqrt{1-\nu^2}}.$$
Therefore by Kato's inequality
\begin{equation}
\pscal{\Lambda_0^-V\psi_+, \left(\Lambda_0^-(\sqrt{1-\Delta}+E-V)\Lambda_0^-\right)^{-1}\Lambda_0^-V\psi_+}\lesssim\int_{\R^3}|V||\psi_+|^2
\label{eq:estim_terme_pas_beau}
\end{equation}
and
$$\big|\tilde q_E(\psi_+)\big|\lesssim\pscal{\psi_+,\sqrt{1-\Delta}\,\psi_+}.$$
In addition the map $E\mapsto \tilde q_E(\psi_+)$ is $C^1$ on $(0,\ii)$ with
\begin{equation}
 \frac{\partial}{\partial E}\tilde q_E(\psi_+)=-\int_{\R^3}|\psi_+|^2-\norm{\left(\Lambda_0^-(\sqrt{1-\Delta}+E-V)\Lambda_0^-\right)^{-1}\Lambda_0^-V\psi_+}_{L^2}^2. 
 \label{eq:diff_q_tilde_E}
\end{equation}
Using~\eqref{eq:estim_terme_pas_beau} and the fact that 
$$\left(\Lambda_0^-(\sqrt{1-\Delta}+E-V)\Lambda_0^-\right)^{-2}\leq \frac{1}{E}\left(\Lambda_0^-(\sqrt{1-\Delta}+E-V)\Lambda_0^-\right)^{-1}$$
the right side of~\eqref{eq:diff_q_tilde_E} is well-defined and continuous on $H^{1/2}$. 

In~\cite[Sec.~4.2]{DolEstSer-00} it was proved that
\begin{equation}
\tilde q_E(\psi_+)\geq 0
\label{eq:Hardy_free}
\end{equation}
for all $\sup(V)-1< E\leq\sqrt{1-\nu^2}$ and all $\psi_+\in \Lambda_0^+H^{1/2}(\R^3,\C^4)$. From~\eqref{eq:Hardy_free} we can first deduce that the domain of the quadratic form $\tilde q_E$ is exactly $\Lambda_0^+H^{1/2}$.

\begin{lemma}[The domain of $\tilde q_E$ is $\Lambda_0^+H^{1/2}$ for $\nu<1$]\label{lem:domain_is_H_12}
We have 
\begin{equation}
 \tilde q_E(\psi_+)\geq (1-\nu)^2 \pscal{\psi_+,\sqrt{1-\Delta}\,\psi_+}-4E\nu^2 \|\psi_+\|^2 
 \label{eq:lower_bound_q_tilde}
\end{equation}
for every $\psi_+\in \Lambda^+_0H^{1/2}(\R^3,\C^4)$ and every $E>\max(0,\sup(V)-1)$.
\end{lemma}

The bound~\eqref{eq:lower_bound_q_tilde} can be improved for $\max(0,\sup(V)-1)<E\leq\sqrt{1-\nu^2}$ but it is sufficient for our purposes. From the lemma we obtain that the maximal domain of $\tilde q_E$ is $\Lambda_0^+H^{1/2}(\R^3,\C^4)$, hence $C^\ii_c(\R^3\setminus\{0\},\C^2)$ is dense in this domain. The rest of the proof is then exactly the same as in the Talman case. Note that in the equivalent of Lemma~\ref{lem:value_S}, the corresponding supremum $\tilde S(V,\phi)$ is always finite since the quadratic form is this time defined on $H^{1/2}$.

\medskip

It therefore remains to provide the

\begin{proof}[Proof of Lemma~\ref{lem:domain_is_H_12}]
Using~\eqref{eq:Hardy_free} for $V=-\nu/|x|$ and passing to the limit $\nu\to1$, we get the following Hardy-type inequality~\cite{DolEstSer-00}
\begin{multline}
\pscal{\psi_+,\sqrt{1-\Delta}\,\psi_+}-\int_{\R^3}\frac{|\psi_+|^2}{|x|}\\
+\pscal{\Lambda_0^-\frac{1}{|x|}\psi_+, \left(\Lambda_0^-(\sqrt{1-\Delta}+|x|^{-1})\Lambda_0^-\right)^{-1}\Lambda_0^-\frac{1}{|x|}\psi_+}\geq0
\label{eq:Hardy-type-free}
\end{multline}
for all $\psi_+ \in \Lambda_0^-H^{1/2}$. Now we would like a similar inequality with an additional $E>\max(0,\sup(V)-1)$ in the denominator of the second term. We start by writing
\begin{multline*}
\pscal{\Lambda_0^-\frac{1}{|x|}\psi_+, \left(\Lambda_0^-(\sqrt{1-\Delta}+E+|x|^{-1})\Lambda_0^-\right)^{-1}\Lambda_0^-\frac{1}{|x|}\psi_+}\\
\geq\pscal{\Lambda_0^-\frac{1}{|x|}\psi_+, \left(\Lambda_0^-(\sqrt{1-\Delta}+|x|^{-1})\Lambda_0^-\right)^{-1}\Lambda_0^-\frac{1}{|x|}\psi_+}\\
-E\norm{\left(\Lambda_0^-(\sqrt{1-\Delta}+|x|^{-1})\Lambda_0^-\right)^{-1}\Lambda_0^-\frac{1}{|x|}\psi_+}^2
\end{multline*}
since $(A+E)^{-1}\geq A^{-1}-EA^{-2}$. Now we claim that the operator
\begin{equation}
B_E:=\left(\Lambda_0^-(\sqrt{1-\Delta}+E+|x|^{-1})\Lambda_0^-\right)^{-1}\Lambda_0^-\frac{1}{|x|}
\label{eq:operateur_horrible}
\end{equation}
is bounded as follows:
$$\|B_E\|\leq 2.$$
Using~\eqref{eq:Hardy-type-free}, this eventually implies
\begin{multline}
\pscal{\psi_+,\sqrt{1-\Delta}\,\psi_+}\\
+\pscal{\Lambda_0^-\frac{1}{|x|}\psi_+, \left(\Lambda_0^-(\sqrt{1-\Delta}+E+|x|^{-1})\Lambda_0^-\right)^{-1}\Lambda_0^-\frac{1}{|x|}\psi_+}\\
\geq\int_{\R^3}\frac{|\psi_+|^2}{|x|}-4E\norm{\psi_+}^2.
\label{eq:Hardy-Coulomb_E_general}
\end{multline}

Before providing the proof that $B_E$ in~\eqref{eq:operateur_horrible} is bounded, we  first come back to $\tilde q_E(\psi_+)$. We note that it is a monotone function of the potential $V$. This is perhaps not so obvious from the formula~\eqref{eq:q_tilde_E}, but it becomes clear if we recall that
\begin{multline*}
\tilde q_E(\psi_+)=\sup_{\psi_-\in\Lambda_0^-F_0^-}\bigg\{\pscal{\psi_++\psi_-,D_0(\psi_++\psi_-)}\\+\int_{\R^3}V|\psi_++\psi_-|^2-E(\|\psi_+\|^2+\|\psi_-\|^2)\bigg\}. 
\end{multline*}
So for a lower bound, we may replace $V$ by $-\nu/|x|$ and we obtain, for every $E>0$,
\begin{align*}
 \tilde q_E(\psi_+)\geq& \pscal{\psi_+,\sqrt{1-\Delta}\,\psi_+}-\nu\int_{\R^3}\frac{|\psi_+|^2}{|x|}-E\int_{\R^3}|\psi_+|^2\\
&+\nu^2\pscal{\Lambda_0^-\frac{1}{|x|}\psi_+, \left(\Lambda_0^-(\sqrt{1-\Delta}+E+|x|^{-1})\Lambda_0^-\right)^{-1}\Lambda_0^-\frac{1}{|x|}\psi_+}\\
\geq& (1-\nu^2)\pscal{\psi_+,\sqrt{1-\Delta}\,\psi_+}-\nu(1-\nu)\int_{\R^3}\frac{|\psi_+|^2}{|x|}-4E\nu^2\int_{\R^3}|\psi_+|^2\\
\geq& (1-\nu)\left(1+\nu-\frac{\pi}{2}\nu\right)\pscal{\psi_+,\sqrt{1-\Delta}\,\psi_+}-4E\nu^2\int_{\R^3}|\psi_+|^2.
\end{align*}
In the second inequality we have used~\eqref{eq:Hardy-Coulomb_E_general} and in the last one we have used Kato's inequality~\eqref{eq:Kato}. Using $\pi/2\leq 2$ yields the simpler inequality~\eqref{eq:lower_bound_q_tilde}.

So it remains to prove that $B_E$ in~\eqref{eq:operateur_horrible} is bounded and since
$$B_E=\frac{\Lambda_0^-(\sqrt{1-\Delta}+|x|^{-1})\Lambda_0^-}{\Lambda_0^-(\sqrt{1-\Delta}+E+|x|^{-1})\Lambda_0^-}B_0$$
where the left side has a norm $\leq1$ by the spectral theorem, it suffices to do it for $E=0$. 
We compute
\begin{multline*}
\left\{\Lambda_0^-\left(\sqrt{1-\Delta}+\frac{1}{|x|}\right)\Lambda_0^-\right\}^2
=\Lambda_0^-(1-\Delta)+\left(\Lambda_0^-\frac{1}{|x|}\Lambda_0^-\right)^2\\
+\Lambda_0^-\left(\frac{1}{|x|}\sqrt{1-\Delta}+\sqrt{1-\Delta}\frac{1}{|x|}\right)\Lambda_0^-.
\end{multline*}
It was proved by Lieb in~\cite{Lieb-84} that 
$$\frac{1}{|x|}\sqrt{1-\Delta}+\sqrt{1-\Delta}\frac{1}{|x|}\geq0$$
and therefore we have the operator inequality 
\begin{equation*}
\left\{\Lambda_0^-\left(\sqrt{1-\Delta}+\frac{1}{|x|}\right)\Lambda_0^-\right\}^2
\geq\Lambda_0^-(1-\Delta).
\end{equation*}
The inverse being operator monotone, we deduce that
$$\left\{\Lambda_0^-\left(\sqrt{1-\Delta}+\frac{1}{|x|}\right)\Lambda_0^-\right\}^{-2}\leq \frac{\Lambda_0^-}{1-\Delta}$$
or, equivalently, that 
$$\left\|\left(\Lambda_0^-\left(\sqrt{1-\Delta}+\frac{1}{|x|}\right)\Lambda_0^-\right)^{-1}\Lambda_0^-\sqrt{1-\Delta}\right\|\leq 1.$$
So we can write
$$B_0=\left(\Lambda_0^-(\sqrt{1-\Delta}+|x|^{-1})\Lambda_0^-\right)^{-1}\Lambda_0^-\sqrt{1-\Delta}\frac{1}{\sqrt{1-\Delta}}\frac{1}{|x|}$$
which proves using Hardy's inequality that
$$\|B_0\|\leq \norm{\frac{1}{\sqrt{1-\Delta}}\frac{1}{|x|}}\leq 2.$$
This ends the proof of Lemma~\ref{lem:domain_is_H_12}.
\end{proof}

\section{Proof of Theorem~\ref{thm:quadratic form_domain} on the subcritical domain $\cV$}\label{sec:proof_quadratic form_domain}

\subsection{Proof that $C^\ii_c(\R^3\setminus\{0\},\C^2)$ is dense in $\cV$}

We are going to adapt \cite[proof of Theorem 4.1]{Zhikov-98}. Zhikov considers a scalar function $\phi$, with $|\nabla \phi|^2$ instead of $|\sigma\cdot\nabla\phi|^2$. A crucial step in his proof is to approximate $\phi$ by a function $\phi_\eps$ bounded in a neighbourhood of $0$. This is easily done in his case, just by taking $\phi_\eps= \phi\1(|\phi|\leq\eps^{-1})$, with $\eps$ small. In our case this simple argument fails. Instead we change our unknown and remove the Pauli matrices.

For every $\phi\in L^2(\R^3, \C^2)$, there is a unique $u$ in the homogeneous Sobolev space $\dot{H}^1(\R^3, \C^2)$ such that $\phi= (\sigma\cdot\nabla)u$. Then,
$$ \|\phi\|_{\cV}^2= \int_{\R^3} \frac{|\Delta u|^2}{2-V} + \int_{\R^3}|\nabla u|^2\,.$$
Note that the matrices $\sigma_k$ have disappeared. Now, for $0<\varepsilon <1$ we let $\phi_\eps=(\sigma\cdot\nabla) u_\eps$ where $u_\varepsilon$ is the solution in $\dot{H}^1(\R^3,\C^2)$
of the equation
$$\Delta u_\varepsilon(x)= \1(|x|\geq\eps)\, \Delta u(x)\,.$$
Obviously $\phi_\eps\in \cV$ and
$$\norm{\phi-\phi_\eps}^2_\cV= \int_{B_\eps}\frac{|\Delta u|^2}{2-V}+\int_{\R^3}|\nabla(u-u_\eps)|^2$$
where $B_\eps$ is the ball of radius $\eps$. The first term converges to zero and the second term can be written in the form
\begin{align*}
\int_{\R^3}|\nabla(u-u_\eps)|^2&=-\int_{B_\eps}(u-u_\eps)\overline{\Delta u}\\
&=\frac{1}{4\pi} \int_{B_\varepsilon}\int_{B_\varepsilon}  \frac{\overline{\Delta u(x)} \Delta u(y)}{|x-y|}\,dx\,dy\lesssim\norm{\Delta u}^2_{L^{6/5}(B_\eps)}
\end{align*}
by the Hardy-Littlewood-Sobolev inequality. Now 
$$\norm{\Delta u}^2_{L^{6/5}(B_\eps)}\leq \big\|(2-V)^{-1/2}\Delta u\big\|^2_{L^{2}(B_\eps)}\norm{2-V}^{1/2}_{L^{3/2}(B_\eps)}$$
which tends to zero.

We have proved that $\phi_\eps\to\phi$ strongly in $\cV$. The function $\phi_\eps$ is well behaved close to the origin. Indeed, for each $0<\varepsilon<1$, $u_{\varepsilon}$ is harmonic on $B(0,\varepsilon)$, so there is $M_\varepsilon>0$ such that $|\phi_\eps|\leq |\nabla u_\varepsilon| \leq M_\varepsilon$ on $B_{\eps/2}$.
Then we can follow \cite{Zhikov-98}. For $0<\delta<\varepsilon/2$ we consider the cut-off function $\theta_\delta (x):= \max\big(0, \min( 1, \frac{2|x|}{\delta}-1)\big)$ and let $\phi^\delta_\varepsilon(x)=\theta_\delta(x)\phi_\varepsilon(x)$. We write
\begin{align*}
||\phi_\varepsilon -\phi_\varepsilon^\delta||^2_{\cV}&= \int_{B_\delta}\frac{|(\sigma\cdot\nabla)(1-\theta_\delta)\phi_\varepsilon|^2}{2-V} + (1-\theta_\delta)^2|\phi_\varepsilon|^2\\
&\le 2 \int_{B_\delta}\frac{|(\sigma\cdot\nabla)\phi_\varepsilon|^2}{2-V}+ 2\int_{B_\delta}\frac{|\nabla \theta_\delta|^2|\phi_\varepsilon|^2}{2-V}+ \int_{B_\delta}|\phi_\varepsilon|^2\\
&\le 8 M_\varepsilon^2\frac{4\pi\delta}{3}+\int_{B(0, \delta)}\frac{2 |(\sigma\cdot\nabla)\phi_\varepsilon|^2}{2-V}+|\phi_\varepsilon|^2,
\end{align*}
and for a fixed $\varepsilon>0$, this quantity tends to $0$ as $\delta\to 0$. To end the proof, note that $\phi^\delta_\varepsilon$ vanishes on $B(0, \delta/2)$, so we can regularize it using a convolution product, which ends the proof that $C^\ii_c(\R^3\setminus\{0\},\C^2)$ is dense in $\cV$.

\begin{remark}\label{rmk:simpler_proof_Coulomb}
If we make the further assumption that $V(x)\leq -\eta/|x|$ in a neighborhood of the origin, we can use a much simpler argument. Namely we replace $\phi$ by $\theta_\delta\phi$ with the same $\theta_\delta$ as before and estimate
$$\int_{\R^3}\frac{|\sigma\cdot\nabla(1-\theta_\delta)\phi|^2}{2-V}\leq 2\int_{\R^3}\frac{(1-\theta_\delta)^2|\sigma\cdot\nabla\phi|^2}{2-V}+2\int_{|x|\leq\delta}\frac{(\theta'_\delta)^2|\phi|^2}{2-V}.$$
The first term goes to 0 by the dominated convergence theorem and the second can be bounded by
$$\int_{|x|\leq\delta}\frac{(\theta'_\delta)^2|\phi|^2}{2-V}\lesssim\int_{|x|\leq\delta}\frac{|x|(\theta'_\delta)^2|\phi|^2}{1+|x|}\lesssim\int_{|x|\leq\delta}\frac{|\phi|^2}{|x|}$$
since $|x|\theta_\delta'$ is uniformly bounded.
\end{remark}

\subsection{Proof that $\cV\subset H^{1/2}$}
Using again our assumption that $V$ is bounded from below by the Coulomb potential, we see that
$$\norm{\phi}^2_\cV\geq \int_{\R^3}\frac{|\sigma\cdot \nabla\phi(x)|^2}{2+1/|x|}\,dx.$$
Hence $\phi$ is in $H^1$ outside of the origin, and $|x|^{1/2}\nabla\phi$ is in $L^2$ in a neighborhood of the origin. This turns out to imply that $\phi\in H^{1/2}$, using the following Hardy-type inequality for the part close to the origin.

\begin{lemma}[Another Hardy-type inequality]
We have
\begin{equation}
\int_{\R^3}\big|(-\Delta)^{1/4}\phi(x)\big|^2\,dx\leq \frac\pi2\int_{\R^3}|x|\,|\sigma\cdot \nabla\phi(x)|^2\,dx,
\end{equation}
for every $\phi$ in $C^\ii_c(\R^3\setminus\{0\},\C^2)$.
\end{lemma}
\begin{proof}
Using that $(\sigma\cdot p)^2=|p|^2$ with $p=-i\nabla$, we can write $\phi=|p|^{-2}\sigma\cdot p(\sigma\cdot p)\phi$. Calling $\eta=\sigma\cdot\nabla\phi$, it remains to show the inequality
$$\pscal{\eta,|p|^{-1}\eta}_{L^2}=\norm{|p|^{-3/2}\sigma\cdot p\eta}^2_{L^2}\leq \frac\pi2\int_{\R^3}|x|\,|\eta(x)|^2\,dx$$
which is just Kato's inequality~\eqref{eq:Kato} for $\widehat{\eta}$.
\end{proof}

\section{Proof of Theorem~\ref{thm:quadratic form_domain_critical} on the critical domains $\cW_{\rm C}$ and $\cW$}\label{sec:proof_thm_quadratic form_domain_critical}

\subsection{A pointwise estimate on $\phi_0$ and $\phi_1$}
We start by giving a useful pointwise estimate on $\phi_0$ and $\phi_1$ at the origin.

\begin{lemma}[Pointwise estimates on the spherical averages $\phi_0$ and $\phi_1$]\label{lem:estim_spherical_average}
Let $\phi\in\cW_{\rm C}$ and let $\phi_0=\1_{\{0\}}(\sigma\cdot L)\phi$ and $\phi_1=\sigma\cdot\omega_x\1_{\{-2\}}(\sigma\cdot L)\phi$. Then we have the pointwise estimate
\begin{equation}
\forall r\leq e^{-1},\qquad |\phi_0(r)|+|\phi_1(r)|\lesssim\,\frac{\sqrt{\log(1/r)}}{r}\left(\sqrt{q_0^{\rm C}(\phi)}+\|\phi\|_{L^2}\right).
\label{eq:spherical_phi}
\end{equation}
\end{lemma}
\begin{proof}
Let $v=r\phi_0$ which belongs to $L^2(0,\ii)$ since $\phi_0(|x|)\in L^2(\R^3,\C^2)$. Using Lebesgue's differential theorem, we get, for $0<r<1/2<r'<1$,
\begin{align*}
|v(r)-v(r')|\leq& \left(\int_r^{r'}\frac{s}{1+s}|v'|^2ds\right)^{1/2}\left(\int_r^{r'}\frac{1+s}{s}\,ds\right)^{1/2}\\
\lesssim& \|\phi\|_{\cW_{\rm C}} \sqrt{1-r+\log(1/r)} 
\end{align*}
which gives the result, after integrating over $r'\in(1/2,1)$.
\end{proof}

\subsection{Proof that $C^\ii_c(\R^3\setminus\{0\},\C^2)$ is dense in $\cW_{\rm C}$}
Let $\phi\in\cW_{\rm C}$. For the functions $\phi_+$ and $\phi_-$ we can apply Theorem~\ref{thm:quadratic form_domain} (or even Remark~\ref{rmk:simpler_proof_Coulomb}). Only $\phi_0$ and $\phi_1$ need a new argument. Since the norms are the same for those two, we only deal with $\phi_0$ and call it $\phi$ throughout the proof, for shortness. 

First we approximate $\phi=\phi_0$ by a function supported outside of a neighborhood of the origin. We use 
$\phi_n=\theta_n\phi$
with $\theta_n$ a radial function equal to 0 close to $0$, equal to 1 on $[e^{-1},\ii)$ and which converges to~1 almost surely. 
We have to estimate the norm of $\phi-\phi_n=(1-\theta_n)\phi$, which is
\begin{multline*}
\int_0^\ii\frac{r}{1+r}\Big|(1-\theta_n)(r\phi'+\phi)-r\phi\theta_n'\Big|^2\,dr\\ \leq 2\int_0^\ii\frac{r}{1+r}(1-\theta_n)^2\left|r\phi'+\phi\right|^2\,dr+2\int_0^\ii\frac{r^3\theta_n'(r)^2}{1+r}|\phi|^2\,dr. 
\end{multline*}
The term involving $1-\theta_n$ goes to zero by the dominated convergence theorem. For the second term we cannot use a simple $\theta_n$ such as $\theta(nr)$ because we are lacking estimates on $\phi$. Inserting the bound~\eqref{eq:spherical_phi} gives
\begin{align*}
\int_0^{e^{-1}}\frac{r^3\theta_n'(r)^2}{1+r}|\phi|^2\,dr\lesssim \int_0^{e^{-1}}\theta_n'(r)^2r\log(1/r)\,dr
\end{align*}
which is divergent if we take a function in the form $\theta(nr)$. Using the fact that $(r\log(1/r))^{-1}$ is not integrable at $r=0$, it is possible to construct a $\theta_n$ such that the right side goes to 0.
Let 
\begin{equation}
 \xi_n(r)=\begin{cases}
\frac1n\left(\frac{1}{\alpha_n\log(1/\alpha_n)}-e\right)\frac{r-\alpha_n/2}{\alpha_n}&\text{for $\alpha_n/2\leq r\leq \alpha_n$,}\\
\frac1n\left(\frac{1}{r\log(1/r)}-e\right)&\text{for $\alpha_n\leq r\leq e^{-1}$,}\\
0&\text{for $r\in[0,\alpha_n/2]\cup[e^{-1},\ii)$.}\\
\end{cases}
 \label{eq:def_xi_n}
\end{equation}
where
$\alpha_n=\exp(-e^n)\to0$ is chosen such as to have
$$\log(\log(1/\alpha_n))=\int_{\alpha_n}^{e^{-1}}\frac{ds}{s\log(1/s)}=n.$$
Then we have
\begin{align*}
\int_0^{e^{-1}} \xi_n(r)\,dr=&\frac{\alpha_n}{8n}\left(\frac{1}{\alpha_n\log(1/\alpha_n)}-e\right)+\frac1n\int_{\alpha_n}^{e^{-1}}\frac{1}{r\log(1/r)}\,dr-\frac{1-e\alpha_n}{n}\\
=&1+O(1/n)
\end{align*}
and
\begin{align*}
&\int_0^{e^{-1}} r\log\left(\frac1r\right)\,\xi_n(r)^2\,dr\\
&\qquad\qquad=\frac1{n^2}\left(\frac{1}{\alpha_n\log(1/\alpha_n)}-e\right)^2\alpha_n^2\int_{\frac{1}2}^{1} r\log\left(\frac1{\alpha_n r}\right)\left(r-\frac12\right)^2dr\\
&\qquad\qquad\qquad+\frac1{n^2}\int_{\alpha_n}^{e^{-1}}r\log\left(\frac1r\right)\left(\frac{1}{r\log(1/r)}-e\right)^2\,dr\\
&\qquad\qquad =\frac1{n}+O(1/n^2).
\end{align*}
Therefore we can take
$$\theta_n(r)=\frac{\dps\int_0^r\xi_n(r)\,dr}{\dps\int_0^\ii \xi_n(r)\,dr}.$$

As a last step, since the function $\theta_n\phi$ is now supported outside of a neighborhood of the origin, it can be approximated by functions in $C^\ii_c(\R^3\setminus\{0\},\C^2)$ by usual convolution arguments.

\subsection{Proof that $C^\ii_c(\R^3\setminus\{0\},\C^2)$ is dense in $\cW$}
The proof for an arbitrary potential $V$ is more complicated. Since $\phi_\pm\in \cV$ we can use Theorem~\ref{thm:quadratic form_domain} for those functions and we only have to approximate $\phi_0$ and $\phi_1$. 
Writing $\phi_0=\theta\phi_0+(1-\theta)\phi_0$ for a smooth radial function $\theta$ of compact support, which equals 1 in a neighborhood of 0, we know that $(1-\theta)\phi_0\in H^1\subset\cW$. So we can prove the result for $\phi_0$ supported in, say, the interval $(0,e^{-1})$, an assumption that we make for the rest of the proof. For simplicity of notation we just assume in the rest of the proof that $\phi=\phi(|x|)$ is radial and supported on $(0,e^{-1})$. 
In radial coordinates, our norm is then equivalent to
\begin{equation}
\int_0^{e^{-1}}g(r)r^2\big|\phi'(r)\big|^2dr+\int_0^{e^{-1}}r\big|r\phi'(r)+\phi(r)\big|^2dr+\int_0^{e^{-1}}r^2\big(1+h(r)\big)|\phi(r)|^2dr
\label{eq:norm_critical_radial}
\end{equation}
where
$$g(r)=\frac1{4\pi}\int_{\mathbb{S}^2}\frac{d\omega}{1-V(r\omega)}-\frac{r}{1+r}\geq0,\qquad h(r)=\frac1{4\pi}\int_{\mathbb{S}^2}V(r\omega)\,d\omega+\frac{1}{r}\geq0.$$
The difficulty is of course that we have little information on $g$ and $h$, except from the fact that $g$ and $rh$ are bounded close to 0. 

As a first step we approximate $\phi$ by a function $\phi_\delta$ on which we have more information. Let $0<\delta<e^{-1}$ and $u_\delta$ be the unique solution of the elliptic minimization problem
\begin{multline}
\inf_{u(\delta)=\phi(\delta)}\bigg\{\int_0^{\delta}g(r)r^2\big|u'(r)\big|^2\,dr+\int_0^{\delta}r\big|ru'(r)+u(r)\big|^2\,dr\\+\int_0^{\delta}r^2(1+h(r))|u(r)|^2\,dr\bigg\}.
\label{eq:min_u_delta}
\end{multline}
Multiplying $\phi$ by a phase we can assume that $\phi(\delta)>0$ and then we conclude that $u_\delta\geq0$ on $[0,\delta]$. This is because the functional in the parenthesis decreases when $u$ is replaced by $|u|$. We then let $\phi_\delta=\phi(r)\1(r\geq\delta)+u(r)\1(r\leq\delta)$ which satisfies $\phi_\delta\in\cW$ with 
\begin{multline*}
\norm{\phi-\phi_\delta}_\cW^2\\
\lesssim \int_0^{\delta}g(r)r^2\big|\phi'(r)\big|^2\,dr+\int_0^{\delta}r\big|r\phi'(r)+\phi(r)\big|^2\,dr+\int_0^{\delta}r^2\big(1+h(r)\big)|\phi(r)|^2\,dr. 
\end{multline*}
This tends to zero when $\delta\to0$. 

Next we are going to work with $\phi_\delta$, using the additional properties coming from the fact that $\phi_\delta=u_\delta$ solves the variational problem~\eqref{eq:min_u_delta} on $[0,\delta]$. To shorten our notation, we simply write $u=u_\delta$. The function $u$ solves in a weak sense the degenerate elliptic ordinary differential equation
\begin{equation}
 -\Big(r^2(g(r)+r)u'(r)\Big)'=r(1-r-rh(r))u(r) 
 \label{eq:ODE_u}
\end{equation}
and satisfies the Neumann-type boundary condition that 
$$\lim_{r\to0}r^2(g(r)+r)u'(r)+r^2u(r)=\lim_{r\to0}r^2(g(r)+r)u'(r)=0.$$
Indeed, note that 
\begin{equation}
 |u(r)|\leq C\|\phi\|_\cW\frac{\sqrt{\log(1/r)}}{r} 
 \label{eq:estim_u}
\end{equation}
by Lemma~\ref{lem:estim_spherical_average} since $\phi_\delta\in\cW$, hence $r^2u(r)\to0$ at the origin. Thus, integrating~\eqref{eq:ODE_u} we find that 
\begin{align*}
-r^2(g(r)+r)u'(r)&=\int_0^rs\big(1-s-sh(s)\big)u(s)\,ds\\
&\lesssim \int_0^r\sqrt{\log(1/s)}\,ds=r\sqrt{\log(1/r)}+o(r\sqrt{\log(1/r)}).
\end{align*}
Multiplying by $u(r)\geq 0$ and using~\eqref{eq:estim_u} we find
$$-r^2(g(r)+r)\big(u^2\big)'\lesssim \log(1/r)$$
and therefore
\begin{equation}
 u(r)\lesssim \left(\int_r^\delta\frac{\log(1/s)\,ds}{s^2(s+g(s))}+\phi(\delta)^2\right)^{1/2}\lesssim \left(\int_r^\delta\frac{\log(1/s)\,ds}{s^2(s+g(s))}\right)^{1/2} 
 \label{eq:final_esimate_u}
\end{equation}
for $r\leq\delta/2$. Note that the integral on the right diverges as $r\to0$ since $g$ is bounded and 
\begin{equation}
\int_r^\delta\frac{\log(1/s)\,ds}{s^2(s+g(s))}\geq \frac{1}{\delta+\|g\|_{L^\ii}}\int_r^\delta\frac{\log(1/s)\,ds}{s^2}\sim_{r\to0} \frac{1}{\delta+\|g\|_{L^\ii}}\frac{\log(1/r)}{r}.
\label{eq:estim_below_F}
\end{equation}
The estimate~\eqref{eq:final_esimate_u} is better than~\eqref{eq:estim_u} if $g(r)$ is much larger than $r$ at the origin. For instance when $g$ has a finite limit at $r=0$, we get $\sqrt{\log(1/r)}/\sqrt{r}$ instead of $\sqrt{\log(1/r)}/r$.

Now we follow the proof of the previous section in the Coulomb case. We need to find a sequence $\theta_n$ which is equal to 0 close to $0$, is equal to 1 on $[\delta/2,\ii)$, converges to~1 almost surely, and such that 
$$\lim_{n\to\ii}\int_0^{\delta}(g(r)+r)r^2u(r)^2\theta_n'(r)^2\,dr=0.$$
Plugging our bound~\eqref{eq:final_esimate_u} on $u$, it is sufficient to show that
$$\lim_{n\to\ii}\int_0^{\delta}(g(r)+r)r^2\left(\int_r^\delta\frac{\log(1/s)\,ds}{s^2(s+g(s))}\right)\theta_n'(r)^2\,dr=0.$$
Following the construction~\eqref{eq:def_xi_n} of $\theta_n$ in the previous section, this is possible when
$$\int_0^\delta \frac{dr}{(g(r)+r)r^2\int_r^\delta\frac{\log(1/s)\,ds}{s^2(s+g(s))}}=+\ii.$$
In order to check that this integral is infinite, we introduce for simplicity
$$F(r):=\int_r^\delta\frac{\log(1/s)\,ds}{s^2(s+g(s))}$$
and rewrite
\begin{align*}
 \int_0^\delta \frac{dr}{(g(r)+r)r^2\int_r^\delta\frac{\log(1/s)\,ds}{s^2(s+g(s))}}&=-\int_0^\delta\frac{F'(r)}{F(r)}\,\frac{dr}{\log(1/r)}\\
 &\geq -\frac{\log F(\delta)}{\log(1/\delta)}+\int_0^\delta\frac{\log F(r)}{r\log^2(1/r)}\,dr
\end{align*}
after integrating by parts. From~\eqref{eq:estim_below_F} we obtain $\log F(r)\geq \log(1/r)+o\big(\log(1/r)\big)$ and therefore the integral on the right diverges, as we wanted.

\subsection{Proof that $\cW_{\rm C}\subset H^s(\R^3,\C^2)$ for $0\leq s<1/2$}

We have shown in Theorem~\ref{thm:quadratic form_domain} that $\cV_{\rm C}\subset H^{1/2}(\R^3,\C^2)$, hence $\phi_\pm\in H^{1/2}(\R^3,\C^2)$ and it suffices to show the result for $\phi=\phi_0(|x|)+\sigma\cdot\omega_x\phi_1(|x|)$. In addition, by density we can assume that $\phi_0$ and $\phi_1\in C^\ii_c(0,\ii)$. Again we can prove the result for $\phi_0$ supported in, say, the interval $(0,e^{-1}/2)$, an assumption that we make for the rest of the proof. 

Now it is actually easier to prove that the compactly-supported $\phi_0(|x|)$ belongs to $W^{1,\alpha}(\R^3)$ for every $1\leq \alpha<3/2$, which implies that it belongs to $H^s(\R^3)$ for $0\leq s<1/2$, by the classical Sobolev embeddings. So we have to prove that
$$\int_0^{e^{-1}} r^2|\phi_0'(r)|^\alpha\,dr=\int_0^{e^{-1}} r^{2-\alpha}|r\phi_0'(r)|^\alpha\,dr<\ii.$$
Note that by Lemma~\ref{lem:estim_spherical_average}
$$\int_0^{e^{-1}} r^{2-\alpha}|\phi_0(r)|^\alpha\,dr\lesssim \int_0^{e^{-1}} r^{2(1-\alpha)}|\log(1/r)|^{\alpha/2}\,dr$$
is convergent under the assumption that $\alpha<3/2$. So it suffices to estimate
\begin{multline*}
\int_0^{e^{-1}} r^{2-\alpha}|r\phi_0'(r)+\phi_0(r)|^\alpha\,dr\\
\leq \left(\int_0^{e^{-1}} r^{\frac{4-3\alpha}{2-\alpha}}\,dr\right)^{\frac{2-\alpha}{2}}\left(\int_0^{e^{-1}} r|r\phi_0'(r)+\phi_0(r)|^2\,dr\right)^{\frac\alpha2},
\end{multline*}
where the first integral is again finite when $\alpha<3/2$.

For $\sigma\cdot\omega_x \phi_1(|x|)$ we have 
$$\big|\nabla\sigma\cdot\omega_x \phi_1(|x|)\big|\leq |\phi_1'(|x|)|+\frac{|\phi_1(|x|)|}{|x|}\leq |\phi_1'(|x|)|+C\frac{\sqrt{\log(1/|x|)}}{|x|^2}$$
and the result is the same. This concludes the proof of Theorem~\ref{thm:quadratic form_domain_critical}.\qed

\section{Proof of the resolvent convergence in Theorem~\ref{thm:distinguished-critical}}\label{sec:proof_thm_distinguished_critical}

We assume for simplicity that $V_\eps=\max(V,-1/\eps)$. The proof for the other case $V_\eps=(1-\eps)V$ is very similar. 
Using the min-max formula for the eigenvalues~\cite{DolEstSer-00} and the fact that $q_{0,V_\eps}\geq q_{0,V}\geq q_0^{\rm C}\geq0$, it is known that 
$$\big(\sup(V)-1,\sqrt{1-\nu^2}\big)\cap\sigma(D_{V_\eps})=\emptyset,\qquad \forall 0<\eps<1.$$
The construction of the distinguished self-adjoint extension in~\cite{EstLos-07} actually provides the information that 
$$(\sup(V)-1,\sqrt{1-\nu^2})\cap\sigma(D_{V})=\emptyset$$ 
as well. 
We therefore fix an energy $E\in (\sup(V)-1,\sqrt{1-\nu^2})$ and prove the norm convergence of the resolvent $(D_{V_\eps}-E)^{-1}$ towards $(D_{V}-E)^{-1}$.  By~\cite[Chap.~IV, Sec.~2.6]{Kato} this implies the convergence in norm of $(D_{V_\eps}-z)^{-1}$ towards $(D_{V}-z)^{-1}$ for any $z\notin\sigma(D_V)$.

As a first step we provide a quantitative bound which follows arguments from~\cite{EstLos-07} but is not explicitly written there. Let $f,g\in L^2(\R^3,\C^2)$ be two vectors and $\phi_\eps,\chi_\eps\in H^1(\R^3,\C^2)$ be such that 
$$\big(D_{V_\eps}-E\big)\begin{pmatrix}
\phi_\eps\\ \chi_\eps
\end{pmatrix}=\begin{pmatrix}
f\\g\end{pmatrix},$$
that is,
\begin{equation}
\begin{cases}
(1-E+V_\eps)\phi_\eps-i\sigma\cdot\nabla\chi_\eps=f,\\
(-1-E+V_\eps)\chi_\eps-i\sigma\cdot\nabla\phi_\eps=g.\\
  \end{cases}
  \label{eq:systeme_epsilon}
\end{equation}
Inserting 
$$\chi_\eps=-\frac{1}{1+E-V_\eps}i\sigma\cdot\nabla\phi_\eps-\frac{g}{1+E-V_\eps}$$
we get the equation in $H^{-1}$
\begin{equation}
(1-E+V_\eps)\phi_\eps -\sigma\cdot\nabla\frac{1}{1+E-V_\eps}\sigma\cdot\nabla\phi_\eps=f+i\sigma\cdot\nabla\frac{1}{1+E-V_\eps}g.
\label{eq:equation_phi_eps}
\end{equation}
Integrating against $\phi_\eps$, we find that
$$\int_{\R^3}(1-E+V_\eps)|\phi_\eps|^2+\int_{\R^3}\frac{|\sigma\cdot\nabla\phi_\eps|^2}{1+E-V_\eps}=\int_{\R^3}\phi_\eps^*f-i\int_{\R^3}\frac{\sigma\cdot\nabla\phi_\eps^*}{1+E-V_\eps}g.$$
We can rewrite this in the form
\begin{multline*}
-E\int_{\R^3}|\phi_\eps|^2-E\int_{\R^3}\frac{|\sigma\cdot\nabla\phi_\eps|^2}{(1-V_\eps)(1+E-V_\eps)}+q_{0,V_\eps}(\phi_\eps)\\=\int_{\R^3}\phi_\eps^*f-i\int_{\R^3}\frac{\sigma\cdot\nabla\phi_\eps^*}{1+E-V_\eps}g.
\end{multline*}
From this we conclude that there exists a constant $C$ (depending on $E$ and $V$ but otherwise independent of $\eps$) such that 
$$\|\phi_\eps\|^2_{L^2}+q_{0,V_\eps}(\phi_\eps)+ \int_{\R^3}\frac{|\sigma\cdot\nabla\phi_\eps|^2}{(1+E-V_\eps)^2}\leq C\left(\|f\|_{L^2}^2+\|g\|_{L^2}^2\right).$$
Since $(1+E-V_\eps)^{-1}$ is uniformly bounded, we also have that $\chi_\eps$ is bounded in $L^2$. 
Writing~\eqref{eq:systeme_epsilon} in the form
\begin{equation}
\begin{cases}
(-1-E+V_\eps)\phi_\eps-i\sigma\cdot\nabla\chi_\eps=f-2\phi_\eps,\\
(1-E+V_\eps)\chi_\eps-i\sigma\cdot\nabla\phi_\eps=g+2\chi_\eps.\\
  \end{cases}
  \label{eq:systeme_epsilon_bis}
\end{equation}
we get all the same information with $\phi_\eps$ and $\chi_\eps$ interchanged. In other words, we have shown that the embedding $\cD(D_{V_\eps})\subset \cW\times\cW$ is continuous with a constant independent of $\eps$:
\begin{multline}
\norm{\phi_\eps}_{\cW}+\norm{\chi_\eps}_{\cW}+\norm{\frac{\sigma\cdot\nabla\phi_\eps}{1+E-V_\eps}}_{L^2}+\norm{\frac{\sigma\cdot\nabla\chi_\eps}{1+E-V_\eps}}_{L^2}\\
\leq C\norm{\begin{pmatrix}f\\ g\end{pmatrix}}_{L^2}=C\norm{(D_{V_\eps}-E)\begin{pmatrix}\phi_\eps\\ \chi_\eps\end{pmatrix}}_{L^2},\qquad\forall\phi_\eps,\chi_\eps\in H^1(\R^3,\C^2).
\end{multline}

Now we can pass to the weak limit $\eps\to0$. Since $\cW\subset H^s$ for all $0\leq s<1/2$, we have $\cW\subset L^p$ for $2\leq p<3$ with a locally compact embedding. Hence we can find a subsequence $\eps_n\to0$ such that $\phi_n:=\phi_{\eps_n}\wto\phi\in\cW$ and $\chi_n:=\chi_{\eps_n}\wto\chi$ weakly in $\cW$, weakly in $L^p$ and strongly in $L^p_{\rm loc}$ for every $2\leq p<3$. Passing to the weak limit in~\eqref{eq:systeme_epsilon}, we find 
\begin{equation}
\begin{cases}
(1-E+V)\phi-i\sigma\cdot\nabla\chi=f,\\
(-1-E+V)\chi-i\sigma\cdot\nabla\phi=g.
  \end{cases}
  \label{eq:systeme_limite}
\end{equation}
Since $\Psi=(\phi,\chi)$ is in $L^2$ and satisfies $\phi\in\cW$ and $D_V\Psi\in L^2$, we have $\Psi\in\cD(D_V)$. We know from the selfadjointness of $D_V$ and the fact that $E\notin\sigma(D_V)$~\cite{EstLos-07} that the equation~\eqref{eq:systeme_limite} has a unique solution. Hence the weak limit is independent of the subsequence and we must have $\phi_\eps\wto\phi$ and $\chi_\eps\wto\chi$. This proves the weak convergence of the resolvents. In addition, we have, after passing to the weak limit,
\begin{multline}
\norm{\phi}_{\cW}+\norm{\chi}_{\cW}+\norm{\frac{\sigma\cdot\nabla\phi}{1+E-V}}_{L^2}+\norm{\frac{\sigma\cdot\nabla\chi}{1+E-V}}_{L^2}\\
\leq C\norm{\begin{pmatrix}f\\ g\end{pmatrix}}_{L^2}=C\norm{(D_{V}-E)\begin{pmatrix}\phi\\ \chi\end{pmatrix}}_{L^2}.
\label{eq:estim_embedding_W}
\end{multline}
This tells us that $\cD(D_V)$ is continuously embedded into the spaces corresponding to the norms on the left. This is already present in the proof of~\cite{EstLos-07}, but not explicitly written.

Now we prove the norm convergence of the resolvents. Let $F_\eps=(f_\eps,g_\eps)$ be any sequence in $L^2(\R^3,\C^4)$ such that $\|F_\eps\|^2=\|f_\eps\|^2+\|g_\eps\|^2=1$, $F_\eps\wto F$ weakly in $L^2$ and 
$$\norm{(D_{V_\eps}-E)^{-1}-(D_{V}-E)^{-1}}=\norm{\left((D_{V_\eps}-E)^{-1}-(D_{V}-E)^{-1}\right)F_\eps}.$$
Let then
$$\begin{pmatrix}
\phi_\eps\\ \chi_\eps
\end{pmatrix}=(D_{V_\eps}-E)^{-1}\begin{pmatrix}
f_\eps\\ g_\eps
\end{pmatrix},\qquad \begin{pmatrix}
\phi_\eps'\\ \chi_\eps'
\end{pmatrix}=(D_{V}-E)^{-1}\begin{pmatrix}
f_\eps\\ g_\eps
\end{pmatrix},$$
which implies that
\begin{equation}
\begin{cases}
(1-E+V)(\phi_\eps-\phi'_\eps)-i\sigma\cdot\nabla(\chi_\eps-\chi'_\eps)=(V-V_\eps)\phi_\eps,\\
(-1-E+V)(\chi_\eps-\chi'_\eps)-i\sigma\cdot\nabla(\phi_\eps-\phi'_\eps)=(V-V_\eps)\chi_\eps.
  \end{cases}
  \label{eq:systeme_epsilon_difference}
\end{equation}
From the previous uniform estimates we know that $\phi_\eps$, $\phi'_\eps$, $\chi_\eps$ and $\chi'_\eps$ are uniformly bounded in the norms appearing on the left of~\eqref{eq:estim_embedding_W}. Passing to weak limits as previously, we find that $\phi_\eps-\phi'_\eps\wto\tilde\phi$ and $\chi_\eps-\chi'_\eps\wto\tilde\chi$ weakly with
$$(D_V-E)\begin{pmatrix}\tilde\phi\\\tilde\chi\end{pmatrix}=0$$
and $\tilde\phi\in\cW$, hence $\tilde\phi=\tilde\chi=0$. Our goal is to prove that the convergence is strong in $L^2$. 
Because of the locally compact embedding into $L^2$, it only remains to prove the compactness at infinity. Let then $\theta$ be a smooth radial function which is $0$ in the ball of radius $R$ and 1 outside of the ball of radius $2R$, for any fixed $R>0$. We multiply~\eqref{eq:systeme_epsilon_difference} by $\theta$ and get
\begin{equation}
\begin{cases}
(1-E+V)\theta(\phi_\eps-\phi'_\eps)-i\sigma\cdot\nabla\theta(\chi_\eps-\chi'_\eps)=-i(\chi_\eps-\chi'_\eps)\sigma\cdot\nabla\theta,\\
(-1-E+V)\theta(\chi_\eps-\chi'_\eps)-i\sigma\cdot\nabla\theta(\phi_\eps-\phi'_\eps)=-i(\phi_\eps-\phi'_\eps)\sigma\cdot\nabla\theta.\\
  \end{cases}
  \label{eq:systeme_epsilon_difference_final}
\end{equation}
since $\theta(V-V_\eps)=0$ for $\eps$ small enough (we use here that $V$ can only diverge at the origin). This can be written in the form
\begin{equation}
\big(D_{V}-E\big)\theta\begin{pmatrix}
\phi_\eps-\phi'_\eps\\
\chi_\eps-\chi'_\eps
\end{pmatrix}=-i\theta'\sigma\cdot\omega_x\begin{pmatrix}\chi_\eps-\chi'_\eps\\ \phi_\eps-\phi'_\eps\end{pmatrix}
\end{equation}
where the right side has a compact support, hence converges strongly to $0$ in $L^2$. Since $D_V-E$ is invertible we conclude as we wanted that $\theta(\phi_\eps-\phi_\eps')\to0$ and $\theta(\chi_\eps-\chi_\eps')\to0$ strongly in $L^2$. Together with the locally compact embedding this proves the norm-convergence of the resolvents and ends the proof of Theorem~\ref{thm:distinguished-critical}.\qed

\appendix

\section{Domains of closures in the exact Coulomb case}\label{app:proof-closure}

In this appendix we characterize the domains of the closures of the minimal operators $\dot{D}_{-\nu/r}$ and ${\dot{h}_\nu^\kappa}$. We prove the following

\begin{prop}[Closures of the minimal operators $\dot{D}_{-\nu/r}$ and ${\dot{h}_\nu^\kappa}$] \label{prop:closure}
If 
$$\begin{cases}
\vert\nu\vert\in [0,1]\setminus\{\sqrt{3}/2\}&\text{ for } \kappa=\pm 1,\\
\vert\nu\vert\in [0,1]&\text{ for } |\kappa|\geq2,
\end{cases}$$
then we have 
\begin{equation} \label{eq:domain-closure}
\cD\big(\overline{\dot{h}_{\nu}^{\kappa}}\big)=H^1_0\left((0,\infty),\C^2\right).
\end{equation}
When $\vert\nu\vert =\sqrt{3}/2$ and $\kappa=\pm1$, we have
$$\cD\big(\overline{\dot{h}_\nu^{\pm 1}}\big) \supsetneq H^1_0\left((0,\infty),\C^2\right).$$ 
In addition,
\begin{equation} \label{eq:domain-closure2}
\cD\big(\overline{\dot{D}_{-\nu/r}}\big)=H^1(\R^3,\C^4)\qquad\text{for $\vert\nu\vert\in [0,1]\setminus\{\sqrt{3}/2\}$,}
\end{equation}
and 
$$\cD\big(\overline{\dot{D}_{-\nu/r}}\big)\supsetneq H^1(\R^3,\C^4)\qquad \text{for $\vert\nu\vert =\sqrt{3}/2$.}$$ 
\end{prop}

\begin{proof}
Note that by Hardy's inequality, the operator norms of the Dirac-Coulomb operators $\dot{D}_{-\nu/r}$ and $\dot{h}_{\nu}^{\kappa}$ are controlled by the $H^1$ norms, so the two inclusions 
$$H^1(\R^3,\C^4)\subset\cD\big(\overline{\dot{D}_{-\nu/r}}\big)$$ 
and 
$$H^1_0(0,\infty)\subset\cD\big(\overline{\dot{h}_\nu^{\kappa}}\big)$$ 
are obvious. The only nontrivial conclusions in Proposition \ref{prop:closure} are the reverse inclusions. 

In \cite{LanRej-79,LanRejKla-80} it is proved that $\overline{\dot{h}_{\nu}^{\kappa}}$ is self-adjoint with same domain as $\overline{\dot{h}_{0}^{\kappa}}$, provided $\vert\nu\vert <\sqrt{\kappa^2-1/4}$. Using the identity
$$\int_0^\ii\left|u'(r)+\kappa\frac{u(r)}{r}\right|^2\,dr=\int_0^\ii|u'(r)|^2\,dr+\kappa(\kappa+1)\int_0^\ii\frac{|u(r)|^2}{r^2}\,dr$$
for $u\in C^\ii_c\big((0,\ii),\C^2\big)$, this domain is just found to be $H^1_0((0,\ii),\C^2)$.
In addition, using resolvent estimates, it was proved in \cite{LanRej-79,LanRejKla-80} that
$$\cD(\overline{\dot{D}_{-\nu/\vert x\vert}})=\cD(\overline{\dot{D}_{0}})=H^1(\R^3,\C^4),\qquad\text{for $\vert\nu\vert <\sqrt{3}/2\,.$}$$

When $\vert\nu\vert =\sqrt{3}/2$, it is easy to see that $H^1_0(0,\infty)$ is a strict subspace of $\cD\big(\overline{\dot{h}_\nu^{\pm 1}}\big)$. Indeed, since essential  self-adjointness still holds, $\overline{\dot{h}_\nu^{\pm 1}}$ coincides with the adjoint operator $(\dot{h}_\nu^{\pm 1})^*$. But $\cD\big((\dot{h}_\nu^{\pm 1})^*\big)$ contains a function behaving like $r^{1/2}$ near $0$, and the derivative of such a function cannot be square integrable.

Finally, we study the case of strong fields $\sqrt{3}/2<\vert\nu\vert\leq 1$. For $\vert\nu\vert <\sqrt{15}/2$ one can consider the restriction of $\overline{\dot{D}_{-\nu/\vert x\vert}}$ to the orthogonal space $E_{\vert\kappa\vert\geq 2}$ of the subspace $\ker(K^2-1)$ in $L^2(\R^2,\C^4)$. Then the arguments and estimates of \cite{LanRej-79,LanRejKla-80} immediately imply that this restriction is self-adjoint with domain $H^1(\R^3,\C^4)\cap E_{\vert\kappa\vert\geq 2}$. So
$$\cD(\overline{\dot{D}_{-\nu/\vert x\vert}})\cap E_{\vert\kappa\vert\geq 2}=H^1(\R^3,\C^4)\cap E_{\vert\kappa\vert\geq 2}$$
and we only have to characterize $\cD\big(\overline{\dot{h}_{\nu}^{\pm1}}\big)$ in the case $\sqrt{3}/2<\vert\nu\vert \leq 1$. To our knowledge, this last point is the only novelty of the present Appendix. Our claim is that 
\begin{equation} \label{inclusion}
\cD\big(\overline{\dot{h}_{\nu}^{\pm 1}}\big)\subset H^1_0\big((0,\infty),\C^2\big)\,, \qquad \forall \,\sqrt{3}/2<\vert\nu\vert \leq 1\;.
\end{equation}
Before proving (\ref{inclusion}), let us explain why this inclusion ends the proof of Proposition \ref{prop:closure}. Using formula (\ref{splitting}), together with the classical identity $\vert \nabla\Psi\vert^2=\vert\partial_r\Psi\vert^2+\frac{\vert L\Psi\vert^2}{r^2}$ and the bound
$$\Vert L\Psi(r,\cdot)\Vert^2_{L^2(S^2)}\leq 2 \Vert \Psi(r,\cdot)\Vert^2_{L^2(S^2)}\,,\qquad\forall \Psi\in \ker(K^2-1)\,,$$
we see that $\Vert\nabla\Psi\Vert_{L^2(\R^3)}^2$ is controlled on $\ker(K^2-1)$ by a finite sum of integrals of the forms $\int_0^\infty\vert\frac{d}{dr}(r^{-1}u,r^{-1}v)\vert^2r^2dr$ and $\int_0^\infty\vert(u,v)\vert^2r^{-2}dr$. Using the one-dimensional Hardy inequality, one then finds that such integrals are dominated by $\int_0^\infty\vert\frac{d}{dr}(u,v)\vert^2dr$ on $C^\infty_c(0,\infty)$. So (\ref{inclusion}) implies that
$$\cD\big(\overline{\dot{D}_{-\nu/\vert x\vert}}\big)\cap \ker(K^2-1)\subset H^1(\R^3,\C^4)\,,$$
and finally proves that $\cD\big(\overline{\dot{D}_{-\nu/\vert x\vert}}\big)\subset H^1(\R^3,\C^4)$.

We now prove (\ref{inclusion}). Without any loss in generality, we can assume that $\kappa=1$ and $\sqrt{3}/2<\nu\leq 1$. Indeed, one can change the sign of $\kappa$ by interchanging $u$ and $v$, and the sign of $\nu$ by replacing $(u,v)$ by $(u,-v)$. Let then $(u_n,v_n)$ be a sequence in $C^\infty_c(0,\infty),$ of limit $(u_\infty,v_\infty)$ for the norm of the domain of $\dot{h}_{-\nu/r}^{1}$. Then $(u_n,v_n)\to(u_\infty,v_\infty)$ in $\cap_{\eps>0}H^1(\eps,\infty)$ and
$$(a_n,b_n):=\left(-\frac{\nu u_n}{r}-\frac{d v_n}{dr}+\frac{v_n}{r},-\frac{\nu v_n}{r}+\frac{d u_n}{dr}+\frac{u_n}{r}\right)$$
converges in $L^2(0,\infty)$ to
$$(a_\infty,b_\infty):=\left(-\frac{\nu u_\infty}{r}-\frac{d v_\infty}{dr}+\frac{v_\infty}{r},-\frac{\nu v_\infty}{r}+\frac{d u_\infty}{dr}+\frac{u_\infty}{r}\right)\,.$$
On the other hand, the homogeneous system
$$\left(-\frac{\nu u}{r}-\frac{d v}{dr}+\frac{v}{r},-\frac{\nu v}{r}+\frac{d u}{dr}+\frac{u}{r}\right)=(0,0)$$
admits the solutions $(\nu,1\pm s)r^{\pm s}$ with $s=\sqrt{1-\nu^2}$ when $\nu<1$, and the solutions $(1,1)$, $(\log(r),\log(r)+1)$ when $\nu=1$. So, remembering that $(u_n,v_n)$ vanishes near $0$, the method of variation of the constant gives the formula
\begin{equation} \label{variation-constant}
\begin{pmatrix}
u_n(r)\\
v_n(r)
\end{pmatrix}
=
\begin{pmatrix}
\nu\\
1+s
\end{pmatrix}\int_0^r\left(\frac{\rho}{r}\right)^s \alpha_n(\rho)d\rho +
\begin{pmatrix}
\nu\\
1-s
\end{pmatrix}\int_0^r \left(\frac{r}{\rho}\right)^s \beta_n(\rho)d\rho
\end{equation}
in the case $\nu<1$, with 
$$(\alpha_n,\beta_n)=\frac1{2s}\left(-a_n-b_n\frac{1-s}{\nu},a_n+b_n\frac{1+s}{\nu}\right)$$ 
convergent in $L^2(0,\infty)$. In the case $\nu=1$ the formula is
\begin{equation} \label{variation-nu=1}
\begin{pmatrix}
u_n(r)\\
v_n(r)
\end{pmatrix}
=
\int_0^r\begin{pmatrix}
b_n(\rho)\\
-a_n(\rho)
\end{pmatrix}d\rho +
\int_0^r
\log(\rho/r)\begin{pmatrix}
a_n(\rho)+b_n(\rho)\\
a_n(\rho)+b_n(\rho)
\end{pmatrix}d\rho\;.
\end{equation}
Our last step is to prove the convergence of $(\frac{d u_n}{dr},\frac{d v_n}{dr})$ to $(\frac{d u_\infty}{dr},\frac{d v_\infty}{dr})$ in $L^2(0,1)$. Considering the derivatives in $r$ of formulas
(\ref{variation-constant}) and (\ref{variation-nu=1}), we see that we just need to estimate integrals of the form
$$\int_0^1 \left(r^{s-1}\int_0^r \rho^{-s} F(\rho)d\rho\right)^2dr\qquad\text{and}\qquad \int_0^1 \left(r^{-1}\int_0^r F(\rho)d\rho\right)^2dr$$ 
in terms of $\int_0^1 F^2(r)dr$ for all $F\in L^2(0,1)$.

For the first estimate we take $p>2$ such that $p s<1$ (this is possible since we assume $\sqrt{3}/2<\nu<1$). We denote $q=\frac{p}{p-1}\in (1,2).$ By H\"older's inequality,
\begin{align*}
r^{s-1}\int_0^r \rho^{-s} F(\rho)d\rho &\leq \left( \int_0^r \left(\frac{r}{\rho}\right)^{ps} d(\rho/r)\right)^{1/p}\left( r^{-1}\int_0^r F^q(\rho) d\rho\right)^{1/q}\\
&=(1-ps)^{-1/p}\left( r^{-1}\int_0^r F^q(\rho) d\rho\right)^{1/q}\;.
\end{align*}
Hence, applying the one-dimensional Hardy inequality
$$\int_0^\infty \left( r^{-1}\int_0^r G(\rho)d\rho\right)^{2/q}dr\leq \left(\frac{2}{2-q}\right)^{2/q}\int_0^\infty G^{2/q}(r) dr$$
to $G(r)=F^{q}(r)\1_{0\leq r\leq 1}$, we find
$$\int_0^1 \left(r^{s-1}\int_0^r \rho^{-s} F(\rho)d\rho\right)^2dr \leq (1-ps)^{-2/p}\left(\frac{2}{2-q}\right)^{2/q} \int_0^1 F^2(r)dr\,.$$
This is the needed estimate in the case $\sqrt{3}/2<\nu<1$. The second estimate (needed for $\nu=1$) is much easier. It follows directly from Hardy's inequality
$$\int_0^\infty \left( r^{-1}\int_0^r G(\rho)d\rho\right)^{2}dr\leq 4 \int_0^\infty G^{2}(r) dr$$
applied to $G(r)=F(r)\1_{0\leq r\leq 1}$. This concludes the proof of (\ref{inclusion}), hence of Proposition \ref{prop:closure}.
\end{proof}

\section{Proof of Theorem~\ref{thm:carres_q_C} on $q_\lambda^{\rm C}$ in the critical case}\label{sec:proof_thm_q_C}

In this section we compute the quadratic form $q_\lambda^{\rm C}$ for $V(x)=-|x|^{-1}$, following the method introduced in~\cite{DolEstLosVeg-04}.

\subsection{Computation of $q_0^{\rm C}$}
First we note that the operator $\sigma\cdot L$ commutes with $(\sigma\cdot\nabla) f(|x|)(\sigma\cdot\nabla)$ for any radial function $f$. Indeed, we have 
\begin{equation}
\sigma\cdot L\;\sigma\cdot \nabla+\sigma\cdot \nabla\;\sigma\cdot L=-2\sigma\cdot \nabla.
\label{eq:anti-commutator-L-nabla}
\end{equation}
Using that $\sigma\cdot L$ commutes with scalar radial functions and inserting~\eqref{eq:anti-commutator-L-nabla}, we easily conclude that
\begin{equation}
\left[\sigma\cdot\nabla \frac{1}{1+|x|^{-1}}\sigma\cdot\nabla\,,\, \sigma\cdot L\right] =0.\label{eq:commute}
\end{equation}
Therefore, recalling that $\omega_x=x/|x|$, we have
$$q_0^{\rm C}(\phi)=q_0^{\rm C}(\phi_+)+q_0^{\rm C}(\phi_0)+q_0^{\rm C}(\phi_-)+q_0^{\rm C}\big(\sigma\cdot\omega_x\phi_1\big).$$ 
We compute these four terms separately. We use the formula
\begin{align*}
\big[\sigma\cdot\nabla\;,\; \sigma\cdot h(|x|)x\big]=
h(|x|)+|x|h'(|x|)+2h(|x|)\big(1+\sigma\cdot L\big)
\end{align*}
which, in the particular case $h(r)=1/r$, becomes
\begin{equation}
\left[\sigma\cdot\nabla,\sigma\cdot \frac{x}{|x|}\right]=\frac{2}{|x|}\big(1+\sigma\cdot L\big).
\end{equation}
Denoting
$$f(r)=\frac{r}{1+r}\qquad \text{and}\qquad g(r)=re^r,$$
which satisfy $fg'=g$, we obtain
\begin{align}
\int_{\R^3}\frac{f}{g^2}\Big|\sigma\cdot\nabla(gu)\Big|^2=&\int_{\R^3}\frac{f}{g^2}\Big|g\sigma\cdot\nabla u+u\frac{g'}{r}\sigma\cdot x\Big|^2\nn\\
=&\int_{\R^3}f\Big|\sigma\cdot\nabla u\Big|^2+\int_{\R^3}\frac{|u|^2}{f}-\pscal{u,\left[\sigma\cdot\nabla,\sigma\cdot \frac{x}{|x|}\right]u}\nn\\
=&\int_{\R^3}\frac{|x|}{1+|x|}\Big|\sigma\cdot\nabla u\Big|^2+\int_{\R^3}\left(1-\frac{1}{|x|}\right)|u|^2-2\pscal{u,\frac{\sigma\cdot L}{|x|}u}\nn\\
=&q_0^{\rm C}(u)-2\pscal{u,\frac{\sigma\cdot L}{|x|}u}.
\label{eq:calcul_phi_+}
\end{align}
This gives what we wanted for $u=\phi_+$ and $u=\phi_0$, after computing 
$$\sigma\cdot\nabla(g\phi_+)= g\left(\sigma\cdot\nabla\phi_++\frac{1+|x|}{|x|}\sigma\cdot\omega_x \phi_+\right)$$
and
$$\sigma\cdot\nabla(g\phi_0)= g\,\sigma\cdot\omega_x\left(\phi_0'+\frac{1+|x|}{|x|} \phi_0\right).$$
Similarly, we have
\begin{align}
\int_{\R^3}fg^2\Big|\sigma\cdot\nabla(g^{-1}u)\Big|^2=&\int_{\R^3}fg^2\Big|g^{-1}\sigma\cdot\nabla u-u\frac{g'}{rg^2}\sigma\cdot x\Big|^2\nn\\
=&\int_{\R^3}f\Big|\sigma\cdot\nabla u\Big|^2+\int_{\R^3}\frac{|u|^2}{f}+\pscal{u,\left[\sigma\cdot\nabla,\sigma\cdot \omega_x\right]u}\nn\\
=&q_0^{\rm C}(u)+2\pscal{u,\frac{2+\sigma\cdot L}{|x|}u}
\label{eq:calcul_phi_-}
\end{align}
which gives the result for $\phi_-$, after inserting
$$\sigma\cdot\nabla(g^{-1}\phi_-)= g^{-1}\left(\sigma\cdot\nabla\phi_--\frac{1+|x|}{|x|}\sigma\cdot\omega_x \phi_-\right).$$
For $u=\sigma\cdot\omega_x\phi_1(|x|)$ we have to use in addition that
\begin{equation*}
\sigma\cdot\nabla\sigma\cdot\omega_x\phi_1=\left[\sigma\cdot\nabla,\sigma\cdot\omega_x\right]\phi_1+\phi_1'=\frac{2}{|x|}(1+\sigma\cdot L)\phi_1+\phi_1'=\frac{2}{|x|}\phi_1+\phi_1'.
\end{equation*}
since $\phi_1$ is radial.

\subsection{Simplification of the norm associated with $q_0^{\rm C}$}
In this section we prove that the norm induced by $q_0^{\rm C}$ on $L^2$ is equivalent to the ones given in~\eqref{eq:norm_q_0_simplifiee}.
Let $\phi\in L^2(\R^3,\C^2)$ be such that all the terms in~\eqref{eq:carres_q_0} are finite. First we remark that 
\begin{align*}
\int_{\R^3}\frac{|x|}{1+|x|}\left|\frac{\sigma\cdot x}{|x|^2}(1+|x|)\,\phi_\pm(x)\right|^2\,dx&=\int_{\R^3}\frac{1+|x|}{|x|}|\phi_\pm(x)|^2\,dx
\end{align*}
which is controlled by the $L^2$ norm and by the term involving $\sigma\cdot L$. So we conclude that
$$\int_{\R^3}\frac{|x|}{1+|x|}\big|\sigma\cdot\nabla\phi_\pm(x)\big|^2\,dx<\ii.$$
Using~\eqref{eq:carres_q_0} we have
\begin{equation}
 \int_{\R^3}\frac{|x|}{1+|x|}\big|\sigma\cdot\nabla\phi_+(x)\big|^2\,dx+\|\phi_+\|^2_{L^2}\geq \int_{\R^3}\frac{|\phi_+(x)|^2}{|x|}\,dx+2\pscal{\phi_+,\frac{\sigma\cdot L}{|x|}\phi_+} 
 \label{eq:controle_sigma_L}
\end{equation}
and a similar inequality for $\phi_-$. Therefore there is no need to keep the term involving $\sigma\cdot L$.
For $\phi_0$ and $\phi_1$ we only use the $L^2$ norm to control the terms involving $r\phi_0$ and $r\phi_1$. 

Lastly, we see that the quadratic form
\begin{equation}
\int_{\R^3}\frac{1}{|x|(1+|x|)}\big|\sigma\cdot\nabla|x|\phi(x)\big|^2\,dx
\label{eq:norme_equivalente}
\end{equation}
is also the sum of the similar terms for $\phi_+$, $\phi_-$, $\phi_0$ and $\sigma\cdot \omega_x\phi_1(|x|)$, since $\sigma\cdot L$ commutes with the corresponding operator in the same way as in~\eqref{eq:commute}. 
Therefore the norm associated with $q_0^{\rm C}$ is equivalent in $L^2$ to that given by ~\eqref{eq:norme_equivalente}. However, in practice it will often be more convenient to use the more precise information contained in~\eqref{eq:norm_q_0_simplifiee} for $\phi_0$, $\phi_1$ and $\phi_\pm$.

\subsection{Estimate on $q_\lambda^{\rm C}$ for $\lambda\neq0$}

It is possible to provide a formula for $q_\lambda^{\rm C}(\phi)$ using the two functions
$$f_\lambda(r)=\frac{r}{1+(1+\lambda)r}\qquad\text{and}\qquad g_\lambda(r)=re^{(1+\lambda)r}$$
in~\eqref{eq:calcul_phi_+} and~\eqref{eq:calcul_phi_-}, and the arguments are exactly the same as before. We can also use~\eqref{eq:relation_q_lambda_0_Coulomb} and notice that, for $\lambda>0$,
\begin{align*}
&\lambda\int_{\R^3}\frac{|x|^2|\sigma\cdot\nabla\phi_\pm(x)|^2}{(1+|x|)(1+(1+\lambda)|x|)}\,dx\\
&\leq\lambda(1+\eta)\int_{\R^3}\frac{|\sigma\cdot\nabla|x|\phi(x)|^2}{(1+|x|)(1+(1+\lambda)|x|)}\,dx
+\lambda(1+\eta^{-1})\int_{\R^3}|\phi(x)|^2\,dx\\
&\leq \frac{\lambda(1+\eta)}{1+\lambda}\int_{\R^3}\frac{|\sigma\cdot\nabla|x|\phi(x)|^2}{|x|(1+|x|)}\,dx+\lambda(1+\eta^{-1})\int_{\R^3}|\phi(x)|^2\,dx
\end{align*}
where the coefficient in front of the first integral is $<1$ for $\eta$ small enough.
This concludes the proof of Theorem~\ref{thm:carres_q_C}.\qed

\section{The two-dimensional case}\label{sec:2D}

In two space dimensions, the free Dirac operator 
\begin{equation}
d_0\ = -i\; \sigma_1\partial_1-i\sigma_2\partial_2 + \sigma_3 = \left( \begin{matrix} 1 & -2i\partial_z
\\ -2i\partial_{\bar{z}} & -1 \\ \end{matrix} \right)
\label{def_Dirac2D}
\end{equation}
is self-adjoint in $L^2(\R^2,\C^2)$ with domain $H^1(\R^2,\C^2)$. Here $z=x_1+ix_2$, $\bar{z}=x_1-ix_2$, $\partial_z=\frac{1}{2}(\partial_1-i\partial_2)\,,$ $\partial_{\bar{z}} = \frac{1}{2}(\partial_1+i\partial_2)\,.$
In this section, we consider  Dirac-Coulomb operators of the form 
$$d_V=d_0+V(x)$$ 
where $V(x)$ is a real-valued function satisfying $V(x)\geq -\nu/|x|$, as in three dimensions. The results are very similar to the three-dimensional case, the algebra is simpler and the proofs do not involve any new idea. So we will only state the theorems for completeness, pointing out the main differences. Note that the two-dimensional case is relevant in solid state physics: although the low-energy electronic excitations in graphene are modeled by a massless two-dimensional Dirac equation~\cite{NetGuiPerNovGei-09}, the study of strained graphene involves a massive Dirac operator~\cite{VozKatGui-10}.

We would like to emphasize four main differences of the 2d case, as compared with the 3d case:
\begin{itemize}
 \item The differential operator $-2i\partial_z$ is not formally self-adjoint (contrary to $-i\sigma\cdot\nabla$ in 3d). Its formal adjoint is $-2i\partial_{\bar{z}}$. 
 \item When $V=-\nu/|x|$, the operator $d_V$ is unitarily equivalent to a direct sum of the same radial Dirac operators $h_{\nu}^{\ell}$ as in 3d, but now $\ell$ (which replaces $\kappa$) is an eigenvalue of the orbital momentum operator $L$, taking all relative integer values including $\ell=0$.
 \item The minimal operator $\dot{d}_{-\nu/\vert x\vert}:=d_{-\nu/\vert x\vert}\upharpoonright\C^\infty_c(\R^2\setminus\{0\},\C^2)$ is not essentially self-adjoint for $\nu\neq 0$. In 3d, the operator $\dot{D}_{-\nu/\vert x\vert}$ is essentially self-adjoint when $|\nu|\leq\sqrt3/2$. This difference is due to the presence of the radial Dirac operator $h_{\nu}^{0}$ in the direct sum mentioned above.
 \item The condition for the existence of a unique distinguished self-adjoint extension is $\vert\nu\vert\leq 1/2$ instead of $\vert\nu\vert\leq 1$.
 \item For $0\leq \nu \leq 1/2$ the first eigenvalue of the distinguished extension $d_{-\nu/\vert x\vert}$ is $\sqrt{1-4\nu^2}$ (also eigenvalue of $h_{\nu}^{0}$) instead of $\sqrt{1-\nu^2}$ (eigenvalue of $h_{\nu}^{\pm 1}$) in 3d.
\end{itemize}

With this in mind, one can prove that Formulas (\ref{eq:domain-closure}) and (\ref{eq:domain-closure2}) still hold in 2d for the domains of the closures $\overline{\dot{d}_{-\nu/r}}$ and $\overline{\dot{h}_{\nu}^{\ell}}$ ($\ell \in \Z$), provided $\vert \nu \vert < 1/2$.\medskip

Theorem~\ref{thm:distinguished} stays true in 2d, with appropriate modifications and we do not state it explicitly. In particular we need to ask that $V_2\in L^2(\R^2)$ and there is no equivalent of \textit{(5)}. These results have been mainly proved by Cuenin and Siedentop~\cite{CueSie-14} (see also~\cite{Warmt-11}). In particular, they showed that 
\begin{equation}
\norm{|x|^{-1/2}(D_0+is)^{-1}|x|^{-1/2}}=2,\qquad \forall s\in\R.
\label{eq:Nenciu-Kato2d}
\end{equation}
Here, the norm is 2 instead of 1, this is the reason why the critical coupling parameter is $\nu=1/2$ instead of $\nu=1$.\medskip

The two-dimensional analogue of the Esteban-Loss method for self-adjoint extensions~\cite{EstLos-07,EstLos-08} was discussed in~\cite{MorMul-15,Muller-16,Warmt-11}. As in 3d, we make the stronger assumptions 
\begin{equation}
 \boxed{-\frac{\nu}{|x|}\leq V(x)\quad\text{and}\quad \sup(V)<1+\sqrt{1-4\nu^2}} 
 \label{eq:assumption_V_intro2d}
\end{equation}
for some $0\leq \nu\leq1/2$. Here $\sqrt{1-4\nu^2}$ is the first eigenvalue of the Dirac operator with the Coulomb potential $V_{\rm C}(x)=-\nu/|x|$.
As in three space dimensions, it is important to study the quadratic form
\begin{equation}
q_\lambda(\phi):=4\int_{\R^2}\frac{|\partial_{\bar{z}} \phi(x)|^2}{1-V(x)+\lambda} \,dx+\int_{\R^2}(1+V(x)-\lambda)|\phi(x)|^2\,dx.
 \label{eq:q_E_2D}
\end{equation}
The two-dimensional analogue of the Hardy-type inequality (\ref{eq:Hardy-type2}) is 
\begin{equation}
\int_{\R^2}\frac{4 |\partial_{\bar{z}} \phi(x)|^2}{a+1/2|x|}\,dx+\int_{\R^2}\left(a-\frac{1}{2|x|}\right)|\phi(x)|^2\,dx\geq0
\label{eq:Hardy-type22d}
\end{equation}
for all $a>0$. This inequality was proved recently by M{\"u}ller \cite{Muller-16}, using the indirect method introduced by Dolbeault-Esteban-S\'er\'e \cite{DolEstSer-00} in their proof of (\ref{eq:Hardy-type2}). But a more direct proof can be given by ``completing the square'' in the spirit of~\cite{DolEstLosVeg-04,DolEstDuoVeg-07}, as we will explain later.

Using \eqref{eq:Hardy-type22d} and our assumption that $V$ is bounded from below by the Coulomb potential, we can prove that $q_\lambda+2\lambda\|\phi\|_{L^2}^2\geq0$. In addition, as in three dimensions, it defines a norm which is equivalent to the one given by the quadratic form
\begin{equation}
\norm{\phi}_\cV^2:= \int_{\R^2}\frac{|\partial_{\bar{z}} \phi(x)|^2}{2-V(x)}\,dx+\int_{\R^2}|\phi(x)|^2\,dx.
\label{eq:norm_V2d}
\end{equation}
The corresponding space is, therefore,
\begin{multline}
 \cV=\Big\{\phi\in L^2(\R^2,\C)\cap H^1_{\rm loc}(\R^2\setminus\{0\},\C)\ :\\ (2-V)^{-1/2}\partial_{\bar{z}} \phi\in L^2(\R^2,\C)\Big\}. 
 \label{eq:def_cV2d}
\end{multline}
Later we will state a result saying that $C^\ii_c(\R^2\setminus\{0\},\C)$ is dense in $\cV$ for the norm~\eqref{eq:norm_V2d}, but for shortness we immediately turn to the discussion of  the critical case.

Following ideas from~\cite{DolEstLosVeg-04,DolEstDuoVeg-07} and Appendix~\ref{sec:proof_thm_q_C}, we can provide a more direct proof of (\ref{eq:Hardy-type22d}). It is useful to start with the Coulomb case $V_{\rm C}(x)=-|2 x|^{-1}$, in which case we use the notation
\begin{equation}
q_\lambda^{\rm C}(\phi)=\int_{\R^2}\left\{\frac{8|x|}{1+2(1+\lambda)|x|}\left|\partial_{\bar{z}}\phi\right|^2+\left(1-\lambda-\frac{1}{2|x|}\right)|\phi |^2\right\}\,dx.
\label{eq:def_q_lambda_Coulomb2d}
\end{equation}
We use the orbital momentum operator $L=-i(x_1\partial_2-x_2\partial_1)$. Note that 
$2z\partial_z=(x\cdot\nabla)+L$,  $2\bar{z}\partial_{\bar{z}}=(x\cdot\nabla)-L$. We recall that the set of eigenvalues of $L$ is $\Z$,  and the eigenspace of eigenvalue $l$ consists of functions taking the form $e^{i l\theta}\phi(r)$ in polar coordinates. The following is the analogue of Theorem~\ref{thm:carres_q_C} and its very similar proof will be omitted.

\begin{thm}[Writing $q_\lambda^{\rm C}$ as a sum of squares in 2d]\label{thm:carres_q_C2d}
For every $\phi\in L^2(\R^2,\C)$ we write
$$\phi=\phi_+(x)+\phi_-(x)+\phi_0(r)+e^{-i \theta}\,\phi_1(r)$$
where $\phi_+=\1_{[1,\ii)}( L)\phi$, $\phi_-=\1_{(-\ii,-2]}(L)\phi$, $\phi_0=\1_{\{0\}}( L)\phi$ and $e^{-i \theta}\,\phi_1(r)=\1_{\{-1\}}( L)\phi$. Then
\begin{align}
q_0^{\rm C}(\phi)=&\int_{\R^2}\frac{2|x|}{1+2|x|}\left|2\partial_{\bar{z}}\phi_+ +\frac{(2|x|+1)z}{2|x|^2}\,\phi_+\right|^2\,dx\nn\\
&+\int_{\R^2}\frac{2|x|}{1+2|x|}\left|2\partial_{\bar{z}}\phi_- -\frac{(2|x|+1)z}{2|x|^2}\,\phi_-\right|^2\,dx\nn\\
&+2\pscal{\phi_+,\frac{L}{|x|}\phi_+}+2\pscal{\phi_-,\frac{-1- L}{|x|}\phi_-}\nn\\
&+4\pi\int_0^\ii\frac{1}{1+2r}\left|r\phi_0'(r)+\frac{\phi_0(r)}{2}+r\phi_0(r)\right|^2\,dr\nn\\
&+4\pi\int_0^\ii\frac{1}{1+2r}\left|r\phi_1'(r)+\frac{\phi_1(r)}{2}-r\phi_1(r)\right|^2\,dr\label{eq:carres_q_02d}
\end{align}
for every $\phi\in H^1(\R^2,\C)$. Moreover
\begin{align}
\|\phi\|_{L^2}^2+q_0^{\rm C}(\phi)&\sim \|\phi\|_{L^2}^2+\int_{\R^2}\frac{|x|}{1+|x|}\big|\partial_{\bar{z}}(\phi_++\phi_-)\big|^2\,dx\nn\\
&\quad+\int_{0}^\ii\frac{1}{1+r}\left( \big| r\phi_0'(r)+\frac{\phi_0(r)}{2}\big|^2+\big|r\phi_1'(r)+\frac{\phi_1(r)}{2}\big|^2\right)\,dr\nn\\
&\sim \|\phi\|_{L^2}^2+\int_{\R^2}\frac{\big|\partial_{\bar{z}}|x|^{1/2}\phi(x)\big|^2}{(1+|x|)}\,dx.
\label{eq:norm_q_0_simplifiee2d}
\end{align}
Finally, for all $-1<\lambda<1$, $(2\lambda+1) \|\phi\|_{L^2}^2 +  q_\lambda^{\rm C}(\phi)$ is a positive quadratic form equivalent to $\|\phi\|_{L^2}^2+q_0^{\rm C}(\phi)$.
\end{thm}

The critical spaces in the 2d case are defined similarly as in 3d. In the Coulomb case $V(x)=-|2x|^{-1}$ we introduce
\begin{equation}
\cW_{\rm C}=\bigg\{\phi\in L^2(\R^2,\C)\ :\ \frac{\partial_{\bar{z}} |x|^{1/2}\phi}{(1+|x|)^{1/2}}\in L^2(\R^2,\C)\bigg\}.
\end{equation}
Then we assume that $V(x)\geq -|2x|^{-1}$ and that $\sup(V)<1$. We define the critical space $\cW$ associated with $V$ by
\begin{multline}
\cW=\bigg\{\phi\in \cW_{\rm C} \ :\ \left(\frac1{1-V(x)}-\frac{2|x|}{1+2|x|}\right)^{1/2}\partial_{\bar{z}}\phi\in L^2(\R^3,\C^2),\\
\left(V(x)+\frac{1}{2|x|}\right)^{1/2}\phi\in L^2(\R^3,\C^2)\bigg\}.
\end{multline}

The following is the equivalent of Theorems~\ref{thm:quadratic form_domain} and~\ref{thm:distinguished-critical} in 2d. 

\begin{thm}[The quadratic form domains in 2d]\label{thm:quadratic form_domain_critical2d}
We assume that 
 \begin{equation}
 V(x)\geq -\frac{1}{2|x|}\qquad\text{and}\qquad \sup(V)<2.
 \label{eq:assumption_V_critical2d}
 \end{equation}
Then the space $C^\ii_c(\R^2\setminus\{0\},\C)$ is dense in $\cV$, in $\cW_{\rm C}$ and in $\cW$ for their respective norms. 
In addition we have the continuous embeddings 
$\cV\subset H^{1/2}(\R^2,\C)$
and
$\cW\subset \cW_{\rm C}\subset H^s(\R^2,\C),$
for every $0\leq s<1/2$.
\end{thm}

The proof of Theorem~\ref{thm:quadratic form_domain_critical2d} is very similar to the proofs of Theorems~\ref{thm:quadratic form_domain} and~\ref{thm:quadratic form_domain_critical}. Note however that 
for the density in $\mathcal{V}$, we have to take a different cutoff function: $\theta_\delta(x):=\max(0,1-\log_2(\max(1,\log_\delta |x|))$. Moreover, 
the pointwise estimate on spherical averages of $\phi$ is slightly different in 2d, compared to that in Lemma~\ref{lem:estim_spherical_average}. Instead of \eqref{eq:spherical_phi}, we have
\begin{equation}
\forall r\leq e^{-1},\qquad |\phi_0(r)|+|\phi_1(r)|\lesssim\,\sqrt{\frac{\log(1/r)}{r}}\left(\sqrt{q_0^{\rm C}(\phi)}+\|\phi\|_{L^2}\right).
\label{eq:spherical_phi2d}
\end{equation}

As in 3d, applying the Esteban-Loss method allows to distinguish and define a unique self-adjoint extension from the property that 
$$\cD(d_V)\subset \left\{\Psi=\begin{pmatrix}\phi\\ \chi\end{pmatrix}\in L^2(\R^2,\C^2)\ :\ \phi\in\cV\right\}$$
in the case $0< \nu<1/2$ and 
$$\cD(d_V)\subset \left\{\Psi=\begin{pmatrix}\phi\\ \chi\end{pmatrix}\in L^2(\R^2,\C^2)\ :\ \phi\in\cW\right\}$$
when $\nu=1/2$. For shortness we do not state the equivalent of Theorems~\ref{thm:distinguished} and~\ref{thm:distinguished-critical}. As in Theorem~\ref{thm:distinguished-critical} we can prove the convergence of the resolvents in norm in the 2d case, by following the proof given in Section~\ref{sec:proof_thm_distinguished_critical}. In the subcritical case, as in Corollary~\ref{cor:chi_in_cV} one can infer some information on $\chi$ under the assumption that $d_V\Psi\in L^2(\R^2,\C)$ and that $\phi\in\cV$. However, due to the fact that the adjoint of $i\partial_z$ is $i\partial_{\bar z}$, the proper conclusion is that
$$\cD(d_V)\subset \cV\times\overline{\cV},\qquad\text{for $0<\nu<1/2$.}$$

We conclude with the min-max characterization of eigenvalues in the spectral gap. As in 3d, we denote  $\Lambda_T^+$ (resp. $\Lambda_T^-$) the Talman projectors corresponding to the Talman decomposition 
$$\Psi=\begin{pmatrix}\phi\\ \chi\end{pmatrix}=\begin{pmatrix}\Lambda_T^+\Psi\\ \Lambda_T^-\Psi\end{pmatrix}.$$ 
We also consider the spectral projections
$$\Lambda_0^+=\1(d_0\geq 0),\qquad \Lambda_0^-=\1(d_0\leq 0).$$ 
For a space $F\subseteq H^{1/2}(\R^3,\C^4)$, we consider the min-max levels given by the same formula as \eqref{eq:min-max_principle}, but with $d_V$ instead of $D_V$. We get the same result as in three dimensions, but with the critical value $\nu=1/2$.

\begin{thm}[Min-max formula for eigenvalues in 2d]\label{thm:new2d}  
Let $0<\nu\leq 1/2$. We assume that
 \begin{equation}
 V(x)\geq -\frac{\nu}{|x|}\qquad\text{and}\qquad \sup(V)<1+\sqrt{1-4\nu^2}.
 \label{eq:assumption_V_bis_2D}
 \end{equation}
Let
\begin{equation}
 C^\ii_c(\R^2\setminus\{0\},\C^2)\subseteq F \subseteq H^{1/2}(\R^2,\C^2). 
 \label{eq:condition_F_2D}
\end{equation}
Then, the number $\lambda_{T,F}^{(k)}$ defined in \eqref{eq:min-max_principle}, is independent of the subspace $F$ and coincides with the $k$th eigenvalue of the distinguished self-adjoint extension of $d_V$ larger than or equal to $\sqrt{1-4\nu^2}$, counted with multiplicity (or is equal to $b = \inf \; (\sigma_{\rm ess} (d_V) \cap (\sqrt{1-4\nu^2}, + \infty))$ if there are less than $k$ eigenvalues below $b$). In addition, we have
$$\lambda_{T,F}^{(k)}=\lambda_{0,F}^{(k)}$$
for all $F$ as above and all $k\geq1$.
\end{thm}



\newcommand{\etalchar}[1]{$^{#1}$}


\projects{\noindent M.L. acknowledges financial supports from the European Research Council (Grant Agreements MNIQS 258023 and MDFT 725528)}

\end{document}